\documentclass[a4paper,11pt]{amsart}

\usepackage{graphicx}
\usepackage{amsmath}
\usepackage{amssymb}
\usepackage[active]{srcltx}
\usepackage{hyperref}
\usepackage{bbm}
\usepackage{enumerate}
\usepackage{mathrsfs}

\newtheorem{thm}{Theorem}
\newtheorem{lem}[thm]{Lemma}%
\newtheorem{cor}[thm]{Corollary}%

\newtheorem{Jac}{Jacobi}

\theoremstyle{remark}
\newtheorem{remark}{Remark}[section] %

\theoremstyle{remark}

\numberwithin{equation}{section}

\def\PP{{\mathbb P}}
\def\QQ{{\mathbb Q}}
\def\RR{{\mathbb R}}
\def\TT{{\mathbb T}}

\def\ZZ{{\mathbb Z}}

\def\one{{\mathbbm{1}}}

\def\pr{\operatorname{pr}}

\def\vecb{{\text{\boldmath$b$}}}

\def\vece{{\text{\boldmath$e$}}}
\def\vecf{{\text{\boldmath$f$}}}
\def\vecg{{\text{\boldmath$g$}}}

\def\vecJ{{\text{\boldmath$J$}}}

\def\vecl{{\text{\boldmath$l$}}}
\def\vecell{{\text{\boldmath$\ell$}}}
\def\vecm{{\text{\boldmath$m$}}}

\def\vecq{{\text{\boldmath$q$}}}

\def\vecu{{\text{\boldmath$u$}}}
\def\vecv{{\text{\boldmath$v$}}}

\def\vecw{{\text{\boldmath$w$}}}
\def\vecx{{\text{\boldmath$x$}}}

\def\vecalf{{\text{\boldmath$\alpha$}}}
\def\vecbeta{{\text{\boldmath$\beta$}}}

\def\vecphi{{\text{\boldmath$\phi$}}}

\def\veceta{{\text{\boldmath$\eta$}}}

\def\vectheta{{\text{\boldmath$\theta$}}}
\def\vecxi{{\text{\boldmath$\xi$}}}

\def\US{{\S_1^{d-1}}}

\def\vecL{{\text{\boldmath$L$}}}

\def\vecnull{{\text{\boldmath$0$}}}

\def\curF{{\mathscr F}}

\def\scrB{{\mathcal B}}

\def\scrD{{\mathcal D}}
\def\scrE{{\mathcal E}}

\def\scrL{{\mathcal L}}
\def\scrM{{\mathcal M}}
\def\scrN{{\mathcal N}}

\def\scrP{{\mathcal P}}

\def\scrT{{\mathcal T}}
\def\scrU{{\mathcal U}}

\def\scrX{{\mathcal X}}

\def\fD{{\mathfrak D}}

\def\fR{{\mathfrak R}}

\def\Leb{\operatorname{Leb}}

\def\tvecx{\widetilde{\vecx}}

\def\e{\mathrm{e}}

\def\diag{\operatorname{diag}}

\def\C{\operatorname{C{}}}

\def\L{\operatorname{L{}}}

\def\GL{\operatorname{GL}}

\def\S{\operatorname{S{}}}

\def\SL{\operatorname{SL}}

\def\SO{\operatorname{SO}}

\def\supp{\operatorname{supp}}

\def\trans{^\mathrm{t}}

\def\bs{\backslash}

\def\Onder#1#2#3#4#5{#1 \setbox0=\hbox{$#1$}\setbox1=\hbox{$#2$}
       \dimen0=.5\wd0 \dimen1=\dimen0 \dimen2=\dp0 \dimen3=\dimen2
       \advance\dimen0 by .5\wd1 \advance\dimen0 by -#4
       \advance\dimen1 by -.5\wd1 \advance\dimen1 by -#4
       \advance\dimen2 by -#3 \advance\dimen2 by \ht1
       \advance\dimen2 by 0.3ex \advance\dimen3 by #5
        \kern-\dimen0\raisebox{-\dimen2}[0ex][\dimen3]{\box1}
       \kern\dimen1}

\newcommand{\tvecL}{\widetilde\vecL}

\newcommand{\lsl}{\mathfrak{sl}}

\newcommand{\osigma}{\overline\sigma}
\newcommand{\tA}{\widetilde{A}}
\newcommand{\tJ}{\widetilde{\vecJ}}
\newcommand{\tL}{\widetilde{L}}
\newcommand{\tT}{\widetilde{T}}
\newcommand{\oN}{\overline{N}}

\newcommand{\tR}{\widetilde{\fR}}

\newcommand{\tOmega}{\widetilde\Omega}
\newcommand{\oOmega}{\overline\Omega}

\newcommand{\Q}{\mathbb{Q}}
\newcommand{\R}{\mathbb{R}}
\newcommand{\Z}{\mathbb{Z}}

\newcommand{\col}{\: : \:}

\newcommand{\bn}{\mathbf{0}}

\newcommand{\tphi}{\widetilde\vecphi}

\newcommand{\ttau}{{\widetilde{\tau}}}

\newcommand{\tg}{\widetilde{g}}

\newcommand{\ve}{\varepsilon}

\newcommand{\matr}[4]{\left( \begin{matrix} #1 & #2 \\ #3 & #4 \end{matrix} \right) }
\newcommand{\cmatr}[2]{\left( \begin{matrix} #1 \\ #2 \end{matrix} \right) }

\newcommand{\smatr}[4]{\bigr( \begin{smallmatrix} #1 & #2 \\ #3 & #4 \end{smallmatrix} \bigr) }

\def \todist {\,\,\buildrel\text{\rm d}\over\longrightarrow\,\,}

\def\hole{\rho}

\title[Universal hitting time statistics]{Universal hitting time statistics for integrable flows}
\author{Carl P. Dettmann}
\author{Jens Marklof}
\author{Andreas Str\"ombergsson}
\address{School of Mathematics, University of Bristol,
Bristol BS8 1TW, U.K.\newline
\rule[0ex]{0ex}{0ex} \hspace{8pt}{\tt carl.dettmann@bristol.ac.uk}}
\address{School of Mathematics, University of Bristol,
Bristol BS8 1TW, U.K.\newline
\rule[0ex]{0ex}{0ex} \hspace{8pt}{\tt j.marklof@bristol.ac.uk}}
\address{Department of Mathematics, Box 480, Uppsala University,
SE-75106 Uppsala, Sweden\newline
\rule[0ex]{0ex}{0ex} \hspace{8pt}{\tt astrombe@math.uu.se}}
\date{7 June 2016/9 August 2016. To appear in Journal of Statistical Physics}

\thanks{The research leading to these results has received funding from the European Research Council under the European Union's Seventh Framework Programme (FP/2007-2013) / ERC Grant Agreement n. 291147.  CPD is supported by EPSRC Grant EP/N002458/1. AS is supported by a grant from the G\"oran Gustafsson Foundation for
Research in Natural Sciences and Medicine, and also by the Swedish Research Council Grant 621-2011-3629.}
\dedicatory{To Professors D.\ Ruelle and Ya.G.\ Sinai on the occasion of their 80th birthday}

\begin{document}

\begin{abstract}
The perceived randomness in the time evolution of ``chaotic'' dynamical systems can be characterized by universal probabilistic limit laws, which do not depend on the fine features of the individual system. One important example is the Poisson law for the times at which a particle with random initial data hits a small set. This was proved in various settings for dynamical systems with strong mixing properties. The key result of the present study is that, despite the absence of mixing, the hitting times of integrable flows also satisfy universal limit laws which are, however, not Poisson. We describe the limit distributions for ``generic'' integrable flows and a natural class of target sets, and illustrate our findings with two examples: the dynamics in central force fields and ellipse billiards. The convergence of the hitting time process follows from a new equidistribution theorem in the space of lattices, which is of independent interest. Its proof exploits Ratner's measure classification theorem for unipotent flows, and extends earlier work of Elkies and McMullen.
\end{abstract}

\maketitle

\section{Introduction}

Let $(\scrM,\curF,\nu)$ be a probability space and consider a measure-preserving dynamical system 
\begin{equation}
\varphi^t:\scrM \to \scrM .
\end{equation}
A fundamental question is how often a trajectory with random initial data $x\in\scrM$ intersects a given target set $\scrD\in\curF$ within time $t$. If $\scrD$ is fixed, this problem has led to many important developments in ergodic theory, which show that, if $\varphi^t$ is sufficiently ``chaotic'' (e.g., partially hyperbolic), the number of intersections satisfies a central limit theorem and more general invariance principles. One of the first results in this direction was Sinai's proof of the central limit theorem for geodesic flows \cite{Sinai60} and, with Bunimovich, the finite-horizon Lorentz gas \cite{Bunimovich:1980ur}. We refer the reader to  \cite{Balint:2011jt,Dolgopyat:2009bl,Gouezel15,Szasz:2007uo} for further references to the literature on this subject. In the case of non-hyperbolic dynamical systems, such as horocycle flows or toral translations, the classical stable limit laws generally fail and must be replaced by system-dependent limit theorems \cite{Bufetov13,Bufetov13b,Bufetov14,DF1,DF,Griffin14}. If on the other hand one considers a sequence of target sets $\scrD_\hole\in\curF$ such that $\nu(\scrD_\hole)\to 0$ as $\hole\to 0$, then the number of intersections within time $t$ (now measured in units of the mean return time to $\scrD_\hole$) satisfies a Poisson limit law, provided $\varphi^t$ is mixing with sufficiently rapid decay of correlations. The first results of this type were proved by Pitskel \cite{Pitskel91} for Markov chains, and by Hirata \cite{Hirata93} in the case of Axiom A diffeomorphisms by employing transfer operator techniques and the Ruelle zeta function. (Hirata's paper was in fact motivated by Sinai's work \cite{Sinai91a,Sinai91b} on the Poisson distribution for quantum energy levels of generic integrable Hamiltonians, following a conjecture by Berry and Tabor \cite{Berry77,Marklof01} in the context of quantum chaos.) For more recent studies on the Poisson law for hitting times in ``chaotic'' dynamical systems, see \cite{Abadi11,Chazottes13,Freitas14,Haydn13,Haydn14a,Lucarini16,Rousseau14} and references therein. 

In the present paper we prove  analogous limit theorems for integrable Hamiltonian flows $\varphi^t$, which are not Poisson yet universal in the sense that they do not depend on the fine features of the individual system considered. 
The principal result of this study is explained in Section \ref{sec:Integrable} for the case of flows with two degrees of freedom, where the target set is a union of small intervals of varying position, length and orientation on each Liouville torus. In the limit of vanishing target size, the sequence of hitting times converges to a limiting process which is described in Section \ref{sec:limit}. Sections \ref{sec:central} and \ref{sec:billiards} illustrate the universality of our limit distribution in the case of two classic examples: the motion of a particle in a central force field and the billiard dynamics in an ellipse. In both cases, the limit process for the hitting times, measured in units of the mean return time on each Liouville torus, is independent of the choice of potential or ellipse, and in fact only depends on the number of connected components of the target set on the invariant torus. The results of Section \ref{sec:limit} are generalized in Section \ref{sec:general} to integrable flows with $d$ degrees of freedom, where unions of small intervals are replaced by unions of shrinking dilations of $k$ given target sets. The key ingredient in the proof of the limit theorems for hitting time statistics is the equidistribution of translates of certain submanifolds in the homogeneous space $G/\Gamma$, where $G=\SL(d,\RR)\ltimes(\RR^d)^k$ and $\Gamma=\SL(d,\ZZ)\ltimes(\ZZ^d)^k$. These results, which are stated and proved in Section \ref{HOMDYNsec}, generalize the equidistribution theorems by Elkies and McMullen \cite{Elkies04} in the case of nonlinear horocycles ($d=2$, $k=1$), and are based on Ratner's celebrated measure classification theorem. The application of these results to the hitting times is carried out in Section \ref{MAINPROOFsec}, and builds on our earlier work for the linear flow on a torus \cite{partI}.

\section{Integrable flows with two degrees of freedom}\label{sec:Integrable}

To keep the presentation as transparent as possible, we first restrict our attention to Hamiltonian flows with two degrees of freedom, whose phase space is the four-dimensional symplectic manifold $\scrX$. (The higher dimensional case is treated in Section \ref{sec:general}.) The basic example is of course $\scrX=\RR^2\times\RR^2$, where the first factor represents the particle's position and the second its momentum. To keep the setting more general, we will not assume Liouville-integrability on the entire phase space, but only on an open subset $\scrM\subset\scrX$, a so-called {\em integrable island.} Liouville integrability \cite[Sect.~1.4]{Bolsinov-Fomenko} implies that there is a foliation (the {\em Liouville foliation}) of $\scrM$ by two-dimensional  leaves. Regular leaves are smooth Lagrangian submanifolds of $\scrM$ that fill $\scrM$ bar a set of measure zero. A compact and connected regular leaf is called a {\em Liouville torus}. Every Liouville torus has a neighbourhood that can be parametrised by action-angle variables $(\vectheta,\vecJ)\in\TT^2\times\scrU$, where $\TT^2=\RR^2/\ZZ^2$ and $\scrU$ is a bounded open subset of $\RR^2$. In these coordinates the Hamiltonian flow is given by
\begin{equation}\label{eq:flow}
\varphi^t: \TT^2\times\scrU \to \TT^2\times\scrU, \quad (\vectheta,\vecJ) \mapsto  (\vectheta + t\, \vecf(\vecJ),\vecJ) ,
\end{equation}
with the smooth Hamiltonian vector field $\vecf=\nabla_\vecJ H$. In what follows, the Hamiltonian structure is in fact completely irrelevant, and {\em we will assume $\scrU$ is a bounded open subset of $\R^m$ ($m\geq1$ arbitrary), and $\vecf:\scrU\to\RR^2$ a smooth function.} We will refer to the corresponding $\varphi^t$ in \eqref{eq:flow} simply as an {\em integrable flow}. Even in the Hamiltonian setting, it is often not necessary to represent the dynamics in action-angle variables to apply our theory; cf.\ the examples of the central force field and billiards in ellipses discussed in Sections \ref{sec:central} and \ref{sec:billiards}.

We will consider random initial data $(\vectheta,\vecJ)$ that is distributed according to a given Borel probability measure $\Lambda$ on $\TT^2\times\scrU$. One example is 
\begin{equation}\label{eq:inv}
\Lambda=\Leb_{\TT^2}\times\lambda,
\end{equation}
where $\Leb_{\TT^2}$ is the uniform probability measure on $\TT^2$ and $\lambda$ is a given absolutely continuous Borel probability measure on $\scrU$. This choice of $\Lambda$ is $\varphi^t$-invariant. One of the key features of this work is that our conclusions also hold for more singular and non-invariant measures $\Lambda$, such as $\Lambda=\delta_{\vectheta_0}\times\lambda$, where $\delta_{\vectheta_0}$ is a point mass at $\vectheta_0$.
The most general setting we will consider is to define $\Lambda$ as the push-forward of a given (absolutely continuous) probability measure $\lambda$ on $\scrU$ by the map $\vecJ\mapsto(\vectheta(\vecJ),\vecJ)$,
where $\vectheta:\scrU\to\TT^2$ is a fixed smooth map;
this means that we consider random initial data in $\TT^2\times\scrU$ of the form
$(\vectheta(\vecJ),\vecJ)$, where $\vecJ$ is a random point in $\scrU$ distributed according $\lambda$.
This is the set-up that we use in the formulation of our main result, Theorem \ref{thm:main000} below. We will demonstrate in Remark \ref{rem:the} that this setting is indeed rather general, and allows a greater selection of measures than is apparent; for instance invariant measures of the form \eqref{eq:inv} can be realized within this framework. 

We also note that the smoothness assumptions on $\vecf$ and $\vectheta$ are less restrictive than they may appear: We can allow discontinuities in the derivatives of theses maps, provided there is an open subset $\scrU'\subset\scrU$ with $\lambda(\scrU\setminus\scrU')=0$, so that the restrictions of $\vecf$ and $\vectheta$ to $\scrU'$ are smooth. Furthermore, the smoothness requirements are a result of an application of Sard's theorem in Theorem \ref{KEYEQUIDISTRTHM2} and may in fact be replaced by finite differentiability conditions.

We consider target sets $\scrD_\hole=\scrD_\hole^{(k)}$ that, in each leaf, appear as disjoint unions of $k$ short intervals transversal to the flow direction. To give a precise definition of $\scrD_\hole$, fix smooth functions $\vecu_j:\scrU\to\S^1$, $\vecphi_j:\scrU\to\TT^2$, and $\ell_j:\scrU\to\RR_{>0}$ ($j=1,\ldots,k$) which describe the orientation, midpoint and length of the $j$th interval in each leaf. 
Set
\begin{equation}\label{target00}
\scrD_\hole^{(k)} = \bigcup_{j=1}^k  \scrD(\vecu_j,\vecphi_j,\hole\ell_j) ,
\end{equation}
where
\begin{equation}
\scrD(\vecu,\vecphi,\ell) := \bigg\{ \big(\vecphi(\vecJ)+s  \vecu(\vecJ)^\perp,\vecJ\big)\in\TT^2\times\scrU \,\bigg|\, -\frac{\ell(\vecJ)}{2}<s<\frac{\ell(\vecJ)}{2} \bigg\} ,
\end{equation} 
with $\vecu(\vecJ)^\perp$ denoting a unit vector perpendicular to $\vecu(\vecJ)$.
This yields, in each leaf $\TT^2\times\{\vecJ\}$, a union of $k$ intervals, where the $j$th interval has length $\hole\ell_j(\vecJ)$, is centered at $\vecphi_j(\vecJ)$ and perpendicular to $\vecu_j(\vecJ)$. 
As mentioned, we assume that each interval is transversal to the flow direction,
i.e.\ $\vecu_j(\vecJ)\cdot\vecf(\vecJ)\neq0$ for all $j\in\{1,\ldots,k\}$ and all $\vecJ\in\scrU$;
in fact we will even assume $\vecu_j(\vecJ)\cdot\vecf(\vecJ)>0$, without any 
loss of generality.

Now, for any initial condition $(\vectheta,\vecJ)$, the set of hitting times 
\begin{equation}\label{scrTdef}
\scrT(\vectheta,\vecJ,\scrD_\rho) :=\{ t >0 \mid \varphi^t(\vectheta,\vecJ) \in\scrD_\rho \} 
\end{equation}
is a discrete (possibly empty) subset of $\RR_{>0}$, the elements of which we label by 
\begin{equation}
0<t_1(\vectheta,\vecJ,\scrD_\rho)<t_2(\vectheta,\vecJ,\scrD_\rho)<\ldots .
\end{equation}
We call $t_i(\vectheta,\vecJ,\scrD_\rho)$ the $i$th {\em entry time to $\scrD_\rho$} if $(\vectheta,\vecJ)\notin\scrD_\rho$, and the $i$th {\em return time to $\scrD_\rho$} if $(\vectheta,\vecJ)\in\scrD_\rho$. A simple volume argument (Santalo's formula \cite{Chernov97}) shows that for any fixed $\vecJ\in\scrU$ such that the components of $\vecf(\vecJ)$ are not rationally related,
the first return time to $\scrD_\rho$ on the leaf $\TT^2\times\{\vecJ\}$
satisfies the formula
\begin{align}
\int_{\scrD_\rho}t_1(\vectheta,\vecJ,\scrD_\rho)\,d\nu_\vecJ(\vectheta)=1,
\end{align}
where $\nu_\vecJ$ is the invariant 
measure on $\scrD_\rho$ obtained by disintegrating Lebesgue measure on $\TT^2\times\{\vecJ\}$ 
with respect to the section $\scrD_\rho$ of the flow $\varphi^t$. 
The measure $\nu_\vecJ$ is explicitly given by
\begin{align}
\int_{\scrD_\rho} g \,d\nu_\vecJ=
\sum_{j=1}^k \bigl(\vecu_j(\vecJ)\cdot\vecf(\vecJ)\bigr) \int_{-\rho\ell_j(\vecJ)/2}^{\rho\ell_j(\vecJ)/2} 
g \bigl(\vecphi(\vecJ)+s \vecu(\vecJ)^\perp,\,\vecJ\bigr)\,ds,
\qquad\forall g \in\C(\scrD_\rho).
\end{align}
Recall that by transversality $\vecu_j(\vecJ)\cdot\vecf(\vecJ)>0$.
It follows that the mean return time with respect to $\nu_\vecJ$ equals
\begin{align}\label{mean01}
\frac{\osigma^{(k)}(\vecJ)}{\rho},\qquad\text{where }\hspace{10pt}
\osigma^{(k)} (\vecJ):=   
\frac1{\sum_{j=1}^k \ell_j(\vecJ)\vecu_j(\vecJ)\cdot\vecf(\vecJ)}.
\end{align}
If we also average over $\vecJ$ with respect to the measure $\lambda$,
the mean return time becomes
\begin{align}\label{mean02}
\frac{\osigma^{(k)}_\lambda}{\rho},\qquad\text{where }\hspace{10pt}
\overline \sigma_\lambda^{(k)} :=  \int_\scrU \overline \sigma^{(k)} (\vecJ)\, \lambda(d\vecJ) .
\end{align}
We have assumed here that the pushforward of $\lambda$ by $\vecf$ has no atoms at points with 
rationally related coordinates. This holds in particular if $\lambda$ is $\vecf$-regular as defined below.

For $\vecJ$ a random point in $\scrU$ distributed according $\lambda$,
the hitting times $t_n(\vectheta(\vecJ),\vecJ,\scrD_\rho^{(k)})$ become random variables,
which we denote by $\tau_{n,\rho}^{(k)}$.
Also $\osigma^{(k)}(\vecJ)$ becomes a random variable, which we denote by $\osigma^{(k)}$.
In this paper, we are interested in the distribution of the sequence of entry times 
$\tau^{(k)}_{n,\hole}$ rescaled by the mean return time \eqref{mean02}, or by the conditional mean return time \eqref{mean01}.

Finally we introduce two technical conditions.
Note that 
$\vecf(\vecJ)\neq\bn$ for all $\vecJ\in\scrU$, by the transversality assumption made previously.
We say that $\lambda$ is {\em $\vecf$-regular} if
the pushforward of $\lambda$ under the map
\begin{equation}\label{fregdef}
\scrU \to\S^1, \qquad \vecJ \mapsto \frac{\vecf(\vecJ)}{\|\vecf(\vecJ)\|},
\end{equation}
is absolutely continuous with respect to Lebesgue measure on $\S^1$.
We say a $k$-tuple of smooth functions $\vecphi_1,\ldots,\vecphi_k:\scrU \to\TT^2$ is \textit{$(\vectheta,\lambda)$-generic,}
if for all $\vecm=(m_1,\ldots,m_k)\in\ZZ^k\setminus\{\vecnull\}$ we have
\begin{equation}\label{hyp3expl0}
\lambda\bigg(\bigg\{ \vecJ\in\scrU\col \sum_{j=1}^k m_j \, \big(\vecphi_j(\vecJ)
-\vectheta(\vecJ)\big) \in \RR \vecf(\vecJ) + \QQ^2 \bigg\}\bigg) = 0.
\end{equation}

The following is the main result of this paper.

\begin{thm}\label{thm:main000}
Let $\vecf:\scrU\to \RR^2$ and $\vectheta:\scrU\to\TT^2$ be smooth maps, 
$\lambda$ an absolutely continuous Borel probability measure on $\scrU$,
and for $j=1,\ldots,k$, let $\vecu_j:\scrU \to\S^1$, $\vecphi_j:\scrU \to\TT^2$  and $\ell_j:\scrU\to\RR_{>0}$ be smooth maps.
Assume $\vecu_j(\vecJ)\cdot\vecf(\vecJ)>0$ for all $\vecJ\in\scrU$, $j\in\{1,\ldots,k\}$.
Also assume that $\lambda$ is $\vecf$-regular and $(\vecphi_1,\ldots,\vecphi_k)$ is $(\vectheta,\lambda)$-generic. 
Then there are sequences of random variables $(\tau_i)_{i=1}^\infty$ and $(\widetilde\tau_i)_{i=1}^\infty$ in $\RR_{>0}$ such that in the limit $\hole\to 0$, for every integer $N$,
\begin{equation}\label{thm:main000res1}
\bigg( \frac{\hole \tau_{1,\hole}^{(k)}}{\overline \sigma_\lambda^{(k)}},\ldots,\frac{\hole \tau_{N,\hole}^{(k)}}{\overline \sigma_\lambda^{(k)}} \bigg) \todist (\tau_1,\ldots,\tau_N),
\end{equation}
and
\begin{equation}
\bigg( \frac{\hole \tau_{1,\hole}^{(k)}}{\overline \sigma^{(k)}},\ldots,\frac{\hole \tau_{N,\hole}^{(k)}}{\overline \sigma^{(k)}} \bigg) \todist (\widetilde\tau_1,\ldots,\widetilde\tau_N).
\end{equation}
\end{thm}
Note that if $\osigma_\lambda^{(k)}=\infty$ then \eqref{thm:main000res1} is trivial, with $\tau_i=0$ for all $i$, since $\tau_{i,\hole}^{(k)}<\infty$ a.s.\ for every fixed $\rho$.

\begin{remark}\label{rem:the}
Recall that Theorem \ref{thm:main000} assumes that the initial data is $(\vectheta(\vecJ),\vecJ)$ with $\vecJ\in\scrU$ distributed according to $\lambda$. This seems to exclude natural choices such as invariant measures of the form \eqref{eq:inv}. Let us demonstrate that this is not the case. The setting of Theorem \ref{thm:main000} (as well as its generalisation to arbitrary dimension $d\geq 2$, Theorem \ref{thm:main001} below) in fact permits random initial data $(\vectheta,\vecJ)$ 
distributed according to any probability measure $\Lambda$ on $\TT^d\times\scrU$ of the form
$\Lambda=\iota_*\lambda_0$,
where $\lambda_0$ is an absolutely continuous Borel probability measure on an open subset $\scrU_0\subset\R^{m_0}$ for some $m_0\in\Z^+$, and some smooth map $\iota:\scrU_0\to\TT^d\times\scrU$. Indeed, such $\Lambda$ can be realized within the setting of Theorem \ref{thm:main000} by using 
\begin{equation}
\scrU_0, \quad \vecf_0:=\vecf\circ\pr_2\circ\iota, \quad \vectheta_0:=\pr_1\circ\iota, \quad \lambda_0
\end{equation}
in place of 
\begin{equation}
\scrU, \quad \vecf, \quad \vectheta, \quad \lambda,
\end{equation}
where $\pr_1,\pr_2$ are the projection maps from $\TT^d\times\scrU$ to $\TT^d$ and $\scrU$, respectively.
Of course, for Theorem \ref{thm:main000} to apply we need to assume that $\lambda_0$ is $\vecf_0$-regular,
and that $(\vecphi_1,\ldots,\vecphi_k)$ is $(\vectheta_0,\lambda_0)$-generic.
\end{remark}

\begin{remark}
We describe the limit sequences $(\tau_i)_{i=1}^\infty$ and $(\widetilde\tau_i)_{i=1}^\infty$ in Section \ref{sec:limit}. A particular highlight is that in the case of a single target ($k=1)$, or in the case of multiple targets with the same lengths $\ell_1=\ldots=\ell_k$ and orientiation $\vecu_1=\ldots=\vecu_k$, the distribution of $(\widetilde\tau_i)_{i=1}^\infty$ is {\em universal}. This means that it is independent of the choice of $\scrU$, $\vecf$, $\lambda$, target orientations, positions and sizes.
In fact a weaker form of universality holds also in the general case, 
and for both $(\tau_i)_{i=1}^\infty$ and $(\widetilde\tau_i)_{i=1}^\infty$.
Indeed, let us define the {\em target weight functions} $\vecL=(L_1,\ldots,L_k)$ and $\tvecL=(\tL_1,\ldots,\tL_k)$
from $\scrU$ to $(\R_{>0})^k$, through
\begin{equation}\label{Ljdef}
L_j(\vecJ)=\overline \sigma_\lambda^{(k)}\,\ell_j(\vecJ)\vecu_j(\vecJ)\cdot\vecf(\vecJ)
\end{equation}
and 
\begin{equation}\label{tLjdef}
\widetilde L_j(\vecJ)=\overline \sigma^{(k)}(\vecJ)\,\ell_j(\vecJ)\vecu_j(\vecJ)\cdot\vecf(\vecJ).
\end{equation}
Then the distribution of $(\tau_i)_{i=1}^\infty$ depends on the system data \textit{only} via the distribution of
$\vecL(\vecJ)$ for $\vecJ$ random in $\scrU$ according to $\lambda$,
and similarly $(\ttau_i)_{i=1}^\infty$ depends only on the distribution of $\tvecL(\vecJ)$.
Furthermore, both $(\tau_i)_{i=1}^\infty$ and $(\ttau_i)_{i=1}^\infty$ yield
stationary point processes, i.e.\ the random set of time points $\{\tau_i\}$ has the same distribution as
$\{\tau_i-t\}\cap\R_{>0}$ for every fixed $t\geq0$,
and similarly for $\{\ttau_i\}$
(cf.\ Section \ref{sec:general}).
\end{remark}

\begin{remark}
Theorem \ref{thm:main000} is stated for the convergence of entry time distributions. It is a general fact that the convergence of entry time distributions implies the convergence of return time distributions and vice versa, with a simple formula relating the two \cite{suspension}.
\end{remark}

\section{The limit distribution}\label{sec:limit}

\begin{figure}
\centerline{\includegraphics[width=400pt]{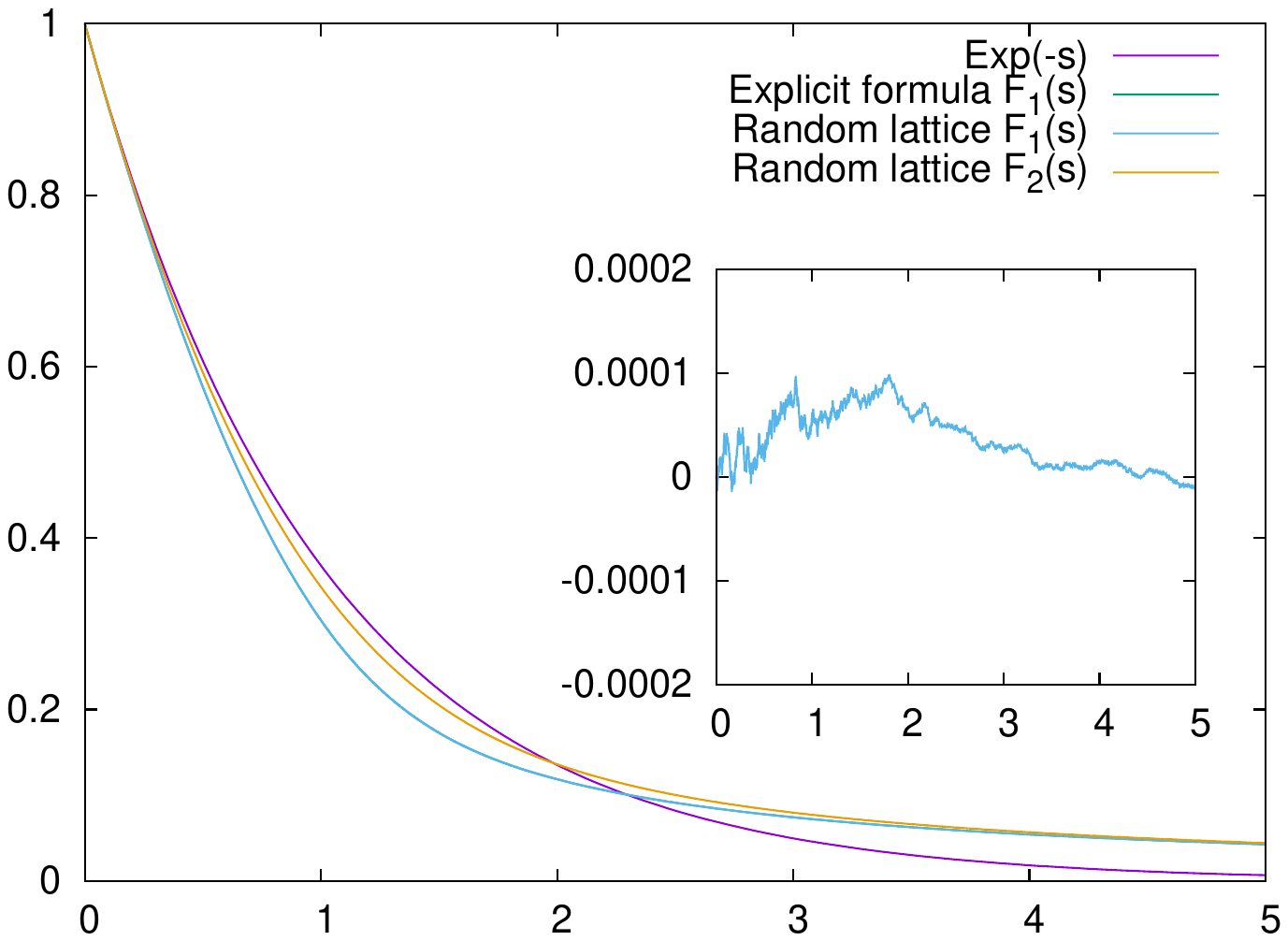}}
\vspace{-25pt}
\caption{Numerically computed $F_1(s)$ and $F_2(s)$, compared with the exponential function $\e^{-s}$ and the explicit formula \eqref{eq:F1s} for $F_1(s)$. The inset shows the difference between the numerically computed $F_1(s)$ and \eqref{eq:F1s}.
\label{Fig1}}
\end{figure}

We will now describe the limit processes $(\tau_i)_{i=1}^\infty$ and $(\widetilde\tau_i)_{i=1}^\infty$ in terms of elementary random variables in the unit cube. A more conceptual description in terms of Haar measure of the special linear group $\SL(2,\RR)$ will be given in Section \ref{sec:general}.

Pick uniformly distributed random points $(a,b,c)$ in the unit cube $(0,1)^3$.
The push-forward of the uniform probability measure under the diffeomorphism
\begin{equation}\label{diffeoo}
(0,1)^3 \to F, \qquad (a,b,c) \mapsto \bigg(\sin(\tfrac\pi 3(a-\tfrac12)) , \frac{\cos(\tfrac\pi 3(a-\tfrac12))}{1-b}  ,\pi c\bigg)
\end{equation}
yields the probability measure $\mu_F=\frac{3}{\pi^2}\, y^{-2} dx\,dy\,d\theta$ on the domain
\begin{align}\label{F}
F=\bigl\{ (x,y,\theta)\in\RR^3 \col |x|<\tfrac12,\: x^2+y^2>1,\: y>0,\: 0<\theta<\pi \bigr\} .
\end{align}
For $x,y,\theta\in\RR$ with $y>0$ and $0\leq\theta<\pi$, consider the Euclidean lattice
\begin{align}
\scrL(x,y,\theta) =  k_{\theta}\matr{\sqrt y}00{1/\sqrt y}\matr10x1\Z^2,
\qquad\text{where }\: k_\theta:=\matr{\cos\theta}{\sin\theta}{-\sin\theta}{\cos\theta} .
\end{align}
A basis for this lattice is given by  the two vectors
\begin{align}\label{LATTICEBASIS}
\vecb_1=y^{-1/2}\,k_{\theta}\cmatr{y}{x}\quad\text{and}\quad
\vecb_2=y^{-1/2}\,k_{\theta}\cmatr{0}{1}.
\end{align}
Note that $\det(\vecb_1,\vecb_2)=1$ and hence $\scrL(x,y,\theta) $ has unit covolume. If we choose $(x,y,\theta)$ random according to the probability measure $\mu_F$, then $\scrL(x,y,\theta)$ represents a {\em random Euclidean lattice} (of covolume one).
Similarly, for $\vecalf\in\TT^2$, the shifted lattice
\begin{align}
\scrL(x,y,\theta,\vecalf) =  k_{\theta}\matr{\sqrt y}00{1/\sqrt y}\matr10x1(\Z^2+\vecalf)
\end{align}
represents a {\em random affine Euclidean lattice} if in addition $\vecalf$ is uniformly distributed in $\TT^2$. 
For a given affine Euclidean lattice $\scrL$ and $\ell>0$,  consider the {\em cut-and-project set}
\begin{equation}\label{PLldef}
\scrP(\scrL, l):= \bigg\{ y_1 >0 : \begin{pmatrix} y_1 \\ y_2 \end{pmatrix} \in \scrL ,\; -\frac{l}{2} < y_2 <\frac{l}{2} \bigg\} \subset \RR_{>0}.
\end{equation}

Let $(x,y,z)$ be randomly distributed according to $\mu_F$, $\vecalf_1,\ldots,\vecalf_k$ be independent and uniformly distributed in $\TT^2$,
and $\vecJ\in\scrU$ distributed according to $\lambda$.
Let $L_j(\vecJ)$ be as in \eqref{Ljdef}.
We will prove in Section \ref{MAINPROOFsec} that 
the elements of the random set
\begin{equation}\label{RANDOMSETKUNION}
\bigcup_{j=1}^k  \scrP\big(\scrL(x,y,\theta,\vecalf_j), L_j(\vecJ)\big),
\end{equation}
ordered by size, form precisely the sequence of random variables $(\tau_i)_{i=1}^\infty$ in Theorem \ref{thm:main000}. This sequence evidently only depends on the choice of target weight function $\vecL$ and the choice of $\scrU$, $\lambda$. 
Similarly, replacing $L_j(\vecJ)$ by $\tL_j(\vecJ)$ (cf.\ \eqref{tLjdef}) in \eqref{RANDOMSETKUNION},
we obtain the sequence $(\widetilde\tau_i)_{i=1}^\infty$.
Note that if $\ell_1=\ldots=\ell_k$ and $\vecu_1=\ldots=\vecu_k$, then $\tL_j(\vecJ)=1/k$,
and thus $(\widetilde\tau_i)_{i=1}^\infty$ is indeed universal as we stated below Theorem \ref{thm:main000}.

Let us describe in some more detail the distribution of the first entry times $\tau_1$ and $\widetilde\tau_1$. In the case of $k$ holes, we have
\begin{equation}
\PP ( \tau_1 > s ) = \int_\scrU F_k(s;\vecL(\vecJ)) \, \lambda(d\vecJ) , 
\end{equation}
\begin{equation}
\PP ( \widetilde \tau_1 > s ) = \int_\scrU  F_k(s;\widetilde \vecL(\vecJ)) \, \lambda(d\vecJ),
\end{equation}
with the universal function
\begin{equation}\label{Fk}
F_k(s,\vecl) = \PP\big( \scrP(\scrL(x,y,\theta,\vecalf_j),l_j)\cap(0,s]=\emptyset \text{ for all $j=1,\ldots,k$} \big) ,
\end{equation}
where $(x,y,\theta)$ is taken to be randomly distributed according to $\mu_F$ and $\vecalf_1,\ldots,\vecalf_k$ independent and uniformly distributed in $\TT^2$, and $\vecl=(l_1,\ldots,l_k)$.
It follows from the invariance properties of the underlying Haar measure (this will become clear in Section \ref{sec:general}) that for any $h>0$
\begin{equation}\label{Fkinv}
F_k\bigg(\frac{s}{h},h\vecl\bigg)= F_k(s,\vecl) . 
\end{equation}

In the case of one hole ($k=1$), the function $F_1(s):=F_1(s,1)$ appears as a limit in various other problems;
notably 
it corresponds to the distribution of free path lengths in the periodic Lorentz gas 
in the small scatterer limit \cite{Boca07,partI}.
It is explicitly given by
\begin{align}\label{eq:F1s}
F_1(s)=\begin{cases}
{\displaystyle\frac 3{\pi^2} s^2 -s+1} & (0 \leq s \leq 1);
\\[11pt]
{\displaystyle
\frac {12}{\pi^2}
( \Xi(s) -  \Xi(s/2))
+ \frac 6{\pi^2} s \log s
+ \Bigl(\frac{6+6 \log 2}{\pi^2} - 2\Bigr)s + \frac{18 \log 2}{\pi^2}}
 & (1 \leq s),
\end{cases}
\end{align}
where $\Xi(s)$ for $s>0$ is defined by
$\Xi''(s) = (1-s^{-1})^2 \log |1-s^{-1}|$ and $\Xi(1)=\Xi'(1)=0$.
In particular $F_1(s)$ has a heavy tail: One has
\begin{align}
F_1(s)=\frac2{\pi^2s}+O\Bigl(\frac1{s^2}\Bigr)
\qquad\text{as }\: s\to\infty.
\end{align}
The formula \eqref{eq:F1s} was derived in 
\cite[Sec.\ 8]{SV}; cf.\ also 
\cite[Theorem 1]{Boca07} and 
\cite{Dahlqvist97}. We are not aware of explicit formulas for the multiple-hole case $k>1$. In this case we evaluate the right hand side of \eqref{Fk} numerically using a Monte Carlo algorithm.
That is, we repeatedly generate a random tuple $(x,y,\theta,\vecalf_1,\ldots,\vecalf_k)$ as described above,
and then determine the smallest $s>0$ such that for some $j\in\{1,\ldots,k\}$
there exists a lattice point $(s,y_2)\in\scrL(x,y,\theta,\vecalf_j)$ in the strip $-l_j/2<y_2<l_j/2$.
In more detail, for given $j$,
in order to determine the left-most point in the intersection of
$\scrL(x,y,\theta,\vecalf_j)$ and the strip $\R_{>0}\times(-l_j/2,l_j/2)$,
one may proceed as follows.
Write $\scrL(x,y,\theta,\vecalf_j)=\vecbeta+\Z\vecb_1+\Z\vecb_2$
with $\vecb_1,\vecb_2$ as in \eqref{LATTICEBASIS} and $\vecbeta\in\R^2$.
After possibly interchanging $\vecb_1$ and $\vecb_2$, and then possibly negating $\vecb_1$,
we may assume that the line $\R\vecb_2$ does not coincide with the $x$-axis
and that the half plane $\R_{>0}\vecb_1+\R\vecb_2$ intersects the $x$-axis in the interval $(0,+\infty)$.
Now determine the smallest integer $m_0$ for which the line $\vecbeta+m_0\vecb_1+\R\vecb_2$ intersects the
strip $\R_{>0}\times(-l_j/2,l_j/2)$,
and then successively for $m=m_0,m_0+1,m_0+2,\ldots$,
check whether there is one or more integers $n$ for which
$\vecbeta+m\vecb_1+n\vecb_2$ lies in the strip.
Note that once this happens for the first time, say for 
$(s',y')=\vecbeta+m_1\vecb_1+n\vecb_2$,
we only need to investigate at most finitely many further $m$-values
$m=m_1+1,m_1+2,\ldots$, namely those for which the line $\vecbeta+m\vecb_1+\R\vecb_2$
intersects the box $(0,s')\times(-l_j/2,l_j/2)$.

Our calculation for $F_2(s):=F_2(s,(\frac12,\frac12))$ used $10^8$ random lattices. The result is presented in Figure \ref{Fig1}. We tested the algorithm by using it to calculate $F_1(s)$ and comparing the resulting graph with the explicit formula \eqref{eq:F1s}.

\section{Central force fields}\label{sec:central}

\begin{figure}
\centerline{\includegraphics[width=400pt]{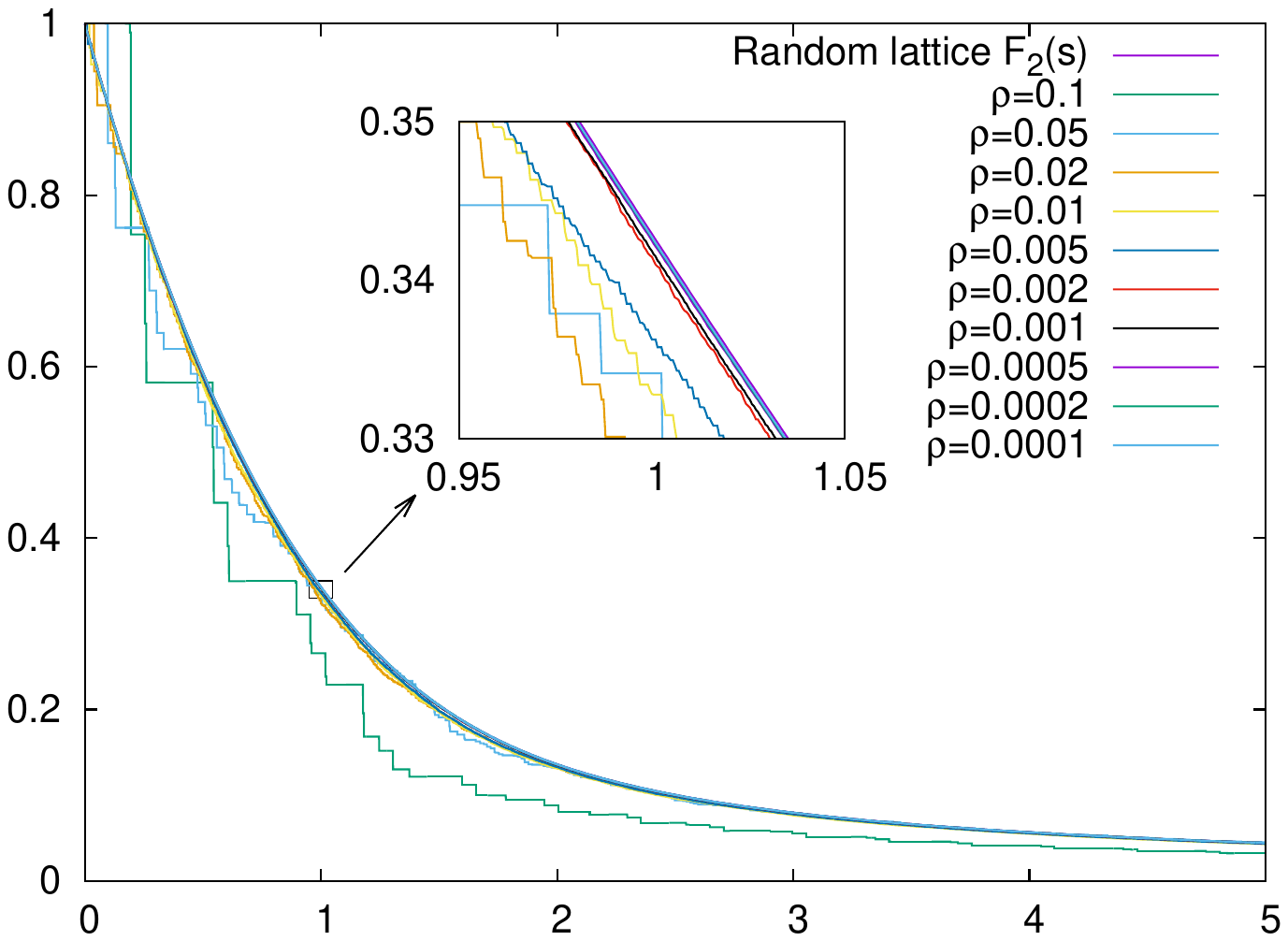}}
\vspace{-25pt}
\caption{Numerical simulations for the entry time distribution $\PP ( \widetilde \tau_1 > s )$ for the potential  $V(r)=r-1$, with different holes sizes $\rho$. We consider particles of mass $m=1$ with initial position in polar coordinates $(r_0,\phi_0)=(1,-2)$, initial velocity $v=0.3$, initial angles uniform in $[0.5,1]$ with a sample size $10^8$.
The target is located at radius $r_0$ and angle interval $[-\rho/2,\rho/2]$. The deviation from the predicted distribution $F_2(s)$ is shown in the inset. 
 \label{f:V}}
\vspace{10pt}
\end{figure} 

\begin{figure}
\centerline{\includegraphics[width=400pt]{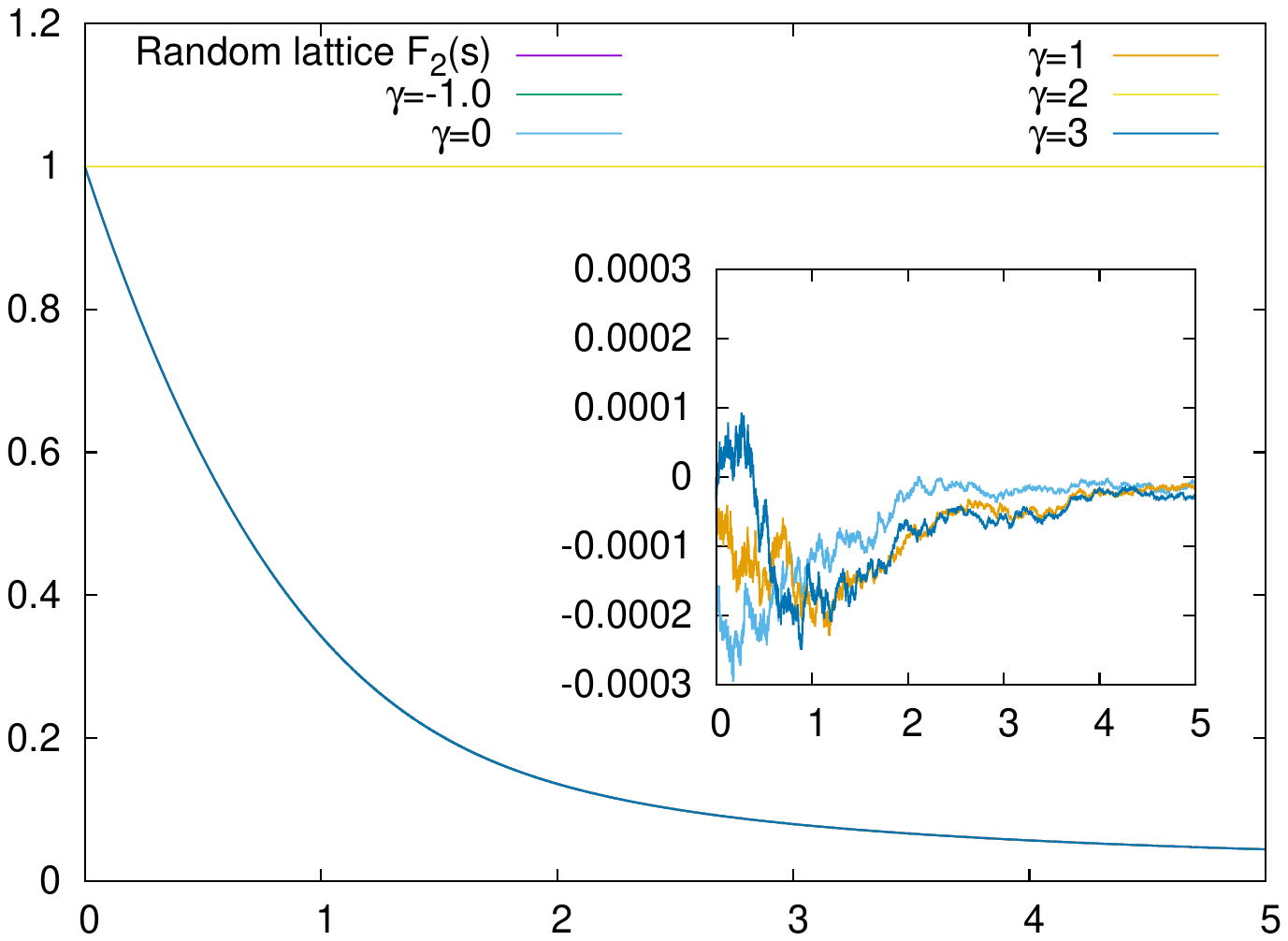}}
\vspace{-25pt}
\caption{Numerical simulations for the entry time distribution $\PP ( \widetilde \tau_1 > s )$ for the potential $V(r)=\frac{r^\gamma-1}{\gamma}$ ($\gamma\neq0$) and $V(r)=\log r$ ($\gamma=0$). The hole size is $\rho=10^{-4}$, and all other parameter values as in Fig.~\ref{f:V}. The cases $\gamma=-1,2$ correspond to the Coulomb potential and isotropic harmonic oscillator, for which the assumptions of Theorem \ref{thm:main000} are not satisfied, and indeed the hitting probability is zero for our choice of initial data. In the remaining cases the deviation from the predicted distribution $F_2(s)$ is shown in the inset. 
\label{f:V2}}
\end{figure}

The dynamics of a point particle subject to central force field in $\RR^3$ takes place in a plane perpendicular to its angular momentum, which is a constant of motion. We choose a coordinate system in which the angular momentum reads $(0,0,L)$, $L\geq 0$. The equations of motion for a particle of unit mass read in polar coordinates
\begin{equation}
\dot\phi = \frac{L}{r^2},  \qquad \dot{r} =\pm \sqrt{2[E-V(r)]-\frac{L^2}{r^2}},
\end{equation}
where $V(r)$ is the potential as a function of the distance to the origin, and $E$ the total energy. It will be convenient to set $\vecJ=(E,L)$, although this choice does {\em not} represent the canonical action variables in this problem. The equations of motion separate, and the dynamics in $r$ is described by a one-dimensional Hamiltonian with effective potential $V(r)+\frac{L^2}{2r^2}$. For a given initial $r_0=r_0(\vecJ)$, the dynamics takes place between the {\em periastron} $r_-=r_-(\vecJ)\leq r_0(\vecJ)$ and the {\em apastron} $r_+=r_+(\vecJ)\geq r_0(\vecJ)$, the minimal/maximal distance to the origin of the particle trajectory with energy $E$ and angular momentum $L$. We will consider cases when the motion is bounded, i.e., $0<r_-\leq r_+<\infty$. Then these values are the turning points of the particle motion, and thus solutions to $V(r)+\frac{L^2}{2r^2}=E$. 
The solution of the equations of motion $(r(t),\phi(t))$ with $(r(0),\phi(0))=(r_0,\phi_0)$ and initial radial velocity $\dot r(0)\geq 0$ is either circular with $\dot r(t)=0$ for all $t$, or otherwise implicitly given by 
\begin{equation}\label{timet}
t = 
\begin{cases}
\displaystyle
\int_{r_0}^{r(t)} \frac{dr'}{\sqrt{2[E-V(r')]-\frac{L^2}{{r'}^2}}}  + n T & (\dot r(t) \geq 0) 
\\[20pt]
\displaystyle
\bigg(\int_{r_0}^{r_+} + \int_{r(t)}^{r_+}\bigg) \frac{dr'}{\sqrt{2[E-V(r')]-\frac{L^2}{{r'}^2}}}  
+ n T & (\dot r(t) \leq 0) ,
\end{cases}
\end{equation}
where $n$ is an arbitrary integer.
The period is 
\begin{equation}\label{Tdef}
T =T(\vecJ) = 2 \int_{r_-(\vecJ)}^{r_+(\vecJ)} \frac{dr}{\sqrt{2[E-V(r)]-\frac{L^2}{r^2}}} .
\end{equation}
Also 
\begin{equation}\label{phit}
\phi(t) = 
\begin{cases}
\displaystyle
\phi_0+\int_{r_0}^{r(t)} \frac{\frac{L}{{r'}^2}\, dr'}{\sqrt{2[E-V(r')]-\frac{L^2}{{r'}^2}}}  + n \alpha & (\dot r(t) \geq 0) 
\\[20pt]
\displaystyle
\phi_0+\bigg(\int_{r_0}^{r_+} + \int_{r(t)}^{r_+}\bigg) \frac{\frac{L}{{r'}^2}\, dr'}{\sqrt{2[E-V(r')]-\frac{L^2}{{r'}^2}}}  
+ n \alpha & (\dot r(t) \leq 0) ,
\end{cases}
\end{equation}
with rotation angle 
\begin{equation}\label{alphaJdef}
\alpha =\alpha(\vecJ) = 2 \int_{r_-(\vecJ)}^{r_+(\vecJ)} \frac{\frac{L}{r^2}\, dr}{\sqrt{2[E-V(r)]-\frac{L^2}{r^2}}}.
\end{equation}

The dynamics is described best by first considering the return map to the cross section defined by 
restricting the radial variable to $r=r_0$ with non-negative radial velocity $\dot r \geq 0$; here $r_0=r_0(\vecJ)$ is permitted to depend on $\vecJ$. This cross section is thus simply parametrized by $\phi \in \RR/2\pi\ZZ$. The corresponding return map is 
\begin{equation}\label{mapp}
\phi \mapsto \phi +\alpha(\vecJ) \bmod 2\pi ,
\end{equation}
with rotation angle $\alpha(\vecJ)$ as in \eqref{alphaJdef}, and return time $T(\vecJ)$ as in \eqref{Tdef}.
We turn the map \eqref{mapp} into a flow of the form \eqref{eq:flow} by considering its suspension flow
\begin{equation}
\varphi^t : \TT^2\times\scrU \to \TT^2\times\scrU,\quad (\vectheta,\vecJ ) \mapsto \bigg(\vectheta+ \frac{t}{T(\vecJ)} \begin{pmatrix} 1 \\ \frac{\alpha(\vecJ)}{2\pi}\end{pmatrix},\vecJ\bigg) .
\end{equation}
A comparison with \eqref{eq:flow} yields 
\begin{equation}
\vecf(\vecJ) = T(\vecJ)^{-1} \begin{pmatrix} 1 \\ \frac{\alpha(\vecJ)}{2\pi}\end{pmatrix} .
\end{equation} 
As to the hypotheses of Theorem \ref{thm:main000},  we see that a Borel probability measure $\lambda$ on
$\scrU$ is {\em $\vecf$-regular} if
the push-forward of $\lambda$ by the map
\begin{equation}\label{f1def}
\scrU \to\RR, \qquad \vecJ \mapsto \alpha(\vecJ),
\end{equation}
is absolutely continuous with respect to Lebesgue measure on $\RR$.
Note that although this condition can hold for most potentials $V$, it fails for the Coulomb potential and the isotropic harmonic oscillator, where every orbit is closed.

A natural choice of target set in polar coordinates is 
\begin{equation}\label{target}
\{ (r,\phi) \mid r=r_0(\vecJ),\; -\pi \hole < \phi < \pi\hole \} ,
\end{equation}
with no restriction on the sign of the radial velocity $\dot r$. We distinguish two cases:

(I) If $r_0(\vecJ)=r_+(\vecJ)$ or $r_0(\vecJ)=r_-(\vecJ)$, the target set is of the form \eqref{target00}, where
\begin{equation}
\scrD_\hole^{(1)} = \scrD\bigg(\vecu_1,\vecphi_1,\hole\bigg) ,
\quad \vecu_1=\begin{pmatrix} 1 \\ 0 \end{pmatrix}, \quad \vecphi_1=\begin{pmatrix} 0 \\ 0 \end{pmatrix}.
\end{equation}
In this simple setting $\vecphi_1=\vecnull$ is $(\vectheta,\lambda)$-generic if (recall \eqref{hyp3expl0}) 
\begin{equation}\label{hyp3explBBB}
\lambda\bigg(\bigg\{ \vecJ\in\scrU \col \vectheta(\vecJ) \in \RR \begin{pmatrix} 1 \\ \frac{\alpha(\vecJ)}{2\pi}\end{pmatrix} + \QQ^2 \bigg\}\bigg) = 0.
\end{equation}

(II) If $r_-(\vecJ) <r_0(\vecJ)< r_+(\vecJ)$, then the particle attains the value $r=r_0(\vecJ)$ with radial velocity $\dot r<0$ before returning to the section $(r_0,\dot r>0)$. The traversed angle is
\begin{equation}
\alpha_*(\vecJ) = 2 \int_{r_0(\vecJ)}^{r_+(\vecJ)} \frac{\frac{L}{r^2}\, dr}{\sqrt{2[E-V(r)]-\frac{L^2}{r^2}}},
\end{equation}
and the corresponding travel time is
\begin{equation}
T_*(\vecJ) = 2 \int_{r_0(\vecJ)}^{r_+(\vecJ)} \frac{dr}{\sqrt{2[E-V(r)]-\frac{L^2}{r^2}}} .
\end{equation}
The target set \eqref{target} has therefore the following angle-action representation, recall \eqref{target00}:
\begin{equation}
\scrD_\hole^{(2)} = \bigcup_{j=1}^2  \scrD(\vecu_j,\vecphi_j,\hole) ,
\end{equation}
with identical orientation
\begin{equation}
\vecu_1(\vecJ)=\vecu_2(\vecJ)=\begin{pmatrix} 1 \\ 0 \end{pmatrix}, 
\end{equation}
located at
\begin{equation}
\vecphi_1(\vecJ)= \begin{pmatrix} 0 \\ 0 \end{pmatrix},\qquad 
\vecphi_2(\vecJ)= \frac{T_*(\vecJ)}{T(\vecJ)} \begin{pmatrix} 1 \\ \frac{\alpha(\vecJ)}{2\pi}\end{pmatrix} - \begin{pmatrix} 0\\  \frac{\alpha_*(\vecJ)}{2\pi} \end{pmatrix} .
\end{equation}
Here the target location is $(\vectheta,\lambda)$-generic if for all $(m_1',m_2')\in\ZZ^2\setminus\{\vecnull\}$
\begin{equation}\label{hyp3explBBB1}
\lambda\bigg(\bigg\{ \vecJ\in\scrU \col m_1' \vectheta(\vecJ) + m_2' \, 
\begin{pmatrix}0\\  \frac{\alpha_*(\vecJ)}{2\pi} \end{pmatrix}
 \in \RR \begin{pmatrix} 1 \\ \frac{\alpha(\vecJ)}{2\pi}\end{pmatrix} + \QQ^2 \bigg\}\bigg) = 0
\end{equation}
(indeed, set $(m_1,m_2)=(m_2'-m_1',-m_2')$ in \eqref{hyp3expl0}).

For our numerical simulations of the first entry time, the relevant parameters used were as follows. The potential is
\begin{equation}
V(r)=\begin{cases}
\frac{r^\gamma-1}{\gamma}&(\gamma\neq 0)\\
\ln(r)& (\gamma=0),
\end{cases}
\end{equation}
where $\gamma\in\RR$, $\gamma>-2$.
The particle mass is $m=1$, initial position in polar coordinates $(r_0,\phi_0)=(1,-2)$, initial velocity $0.3$ with directions uniform in $[0.5,1]\subset[0,2\pi]$ (the sample size is $10^8$); the target is the angular interval $[-\rho/2,\rho/2]$ located at radius $r_0=1$. 
Fig.~\ref{f:V}
displays the results of computations with several values of $\rho$ and fixed $\gamma=1$, and Fig.~\ref{f:V2} the corresponding results for fixed $\rho=10^{-4}$ and various values of $\gamma$.

\section{Integrable billiards}\label{sec:billiards}

\begin{figure}
\centerline{\includegraphics[width=400pt]{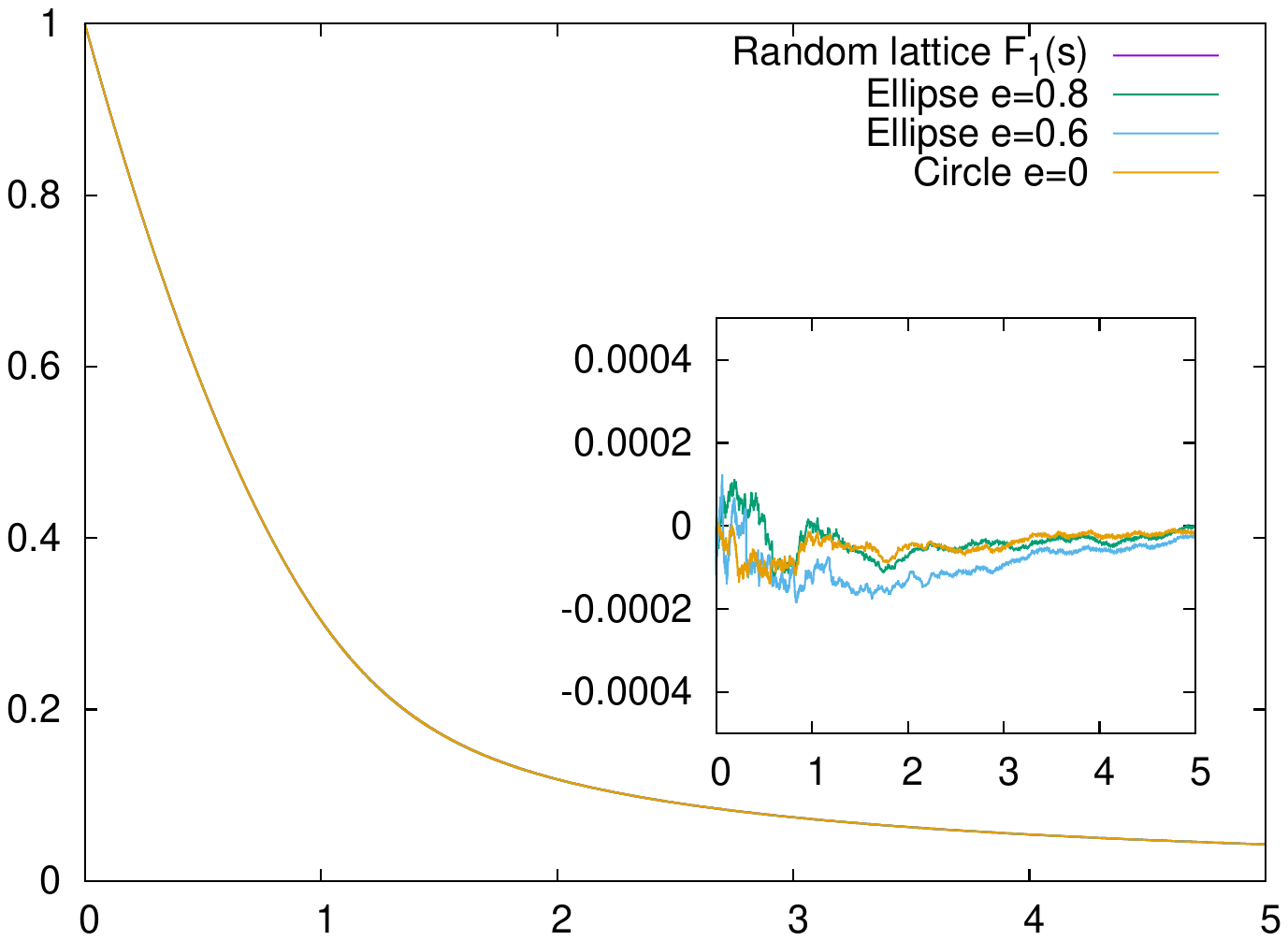}}
\centerline{\includegraphics[width=400pt]{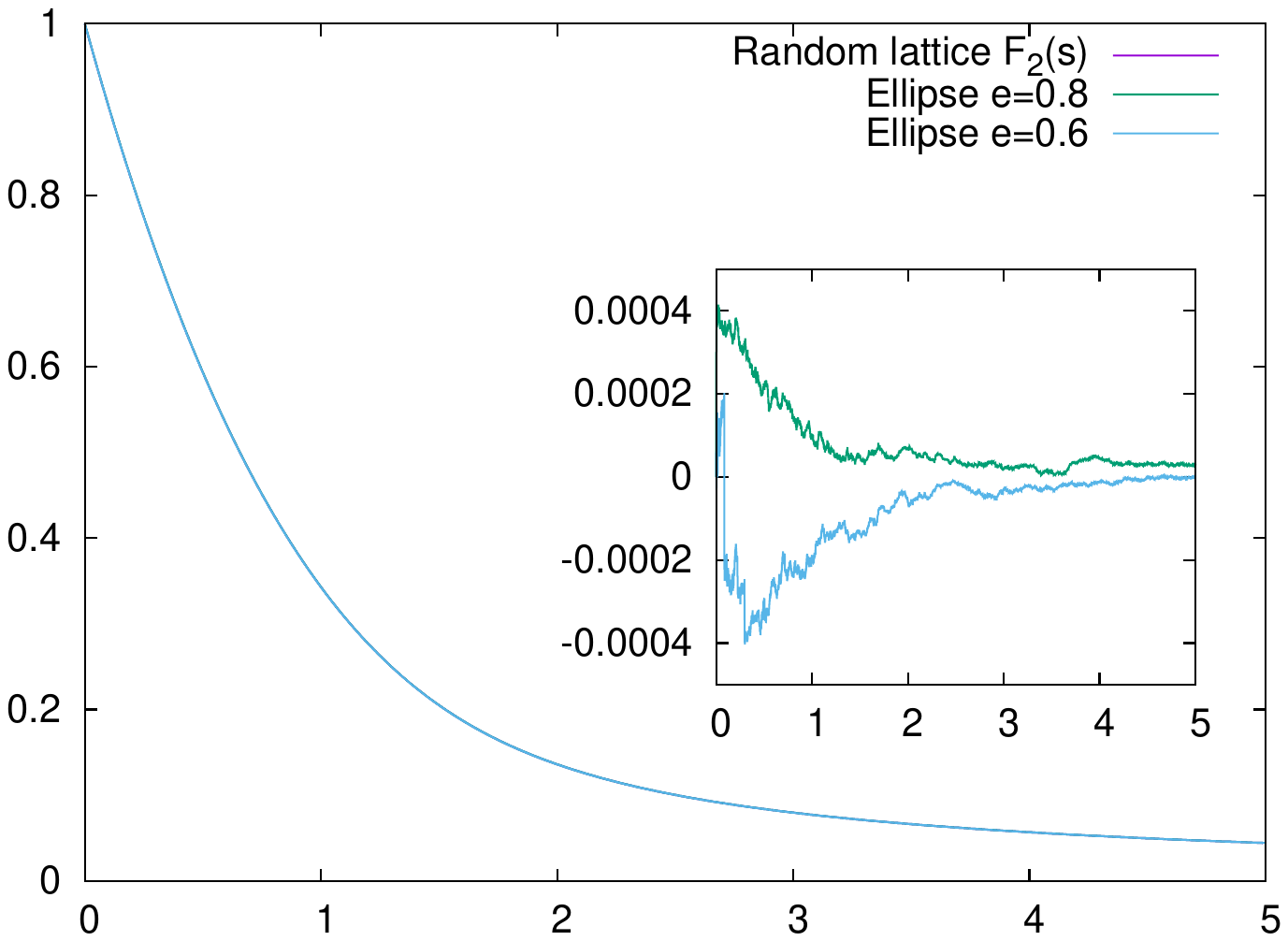}}
\caption{Numerical simulations confirming that the entry time distribution $\PP ( \widetilde \tau_1 > s )$ for an arbitrary ellipse scales to the expected universal functions for initial conditions with $\ve<1$ (upper panel) and $\ve>1$ (lower panel).  The inset panels highlight the difference between the ellipse simulations and theoretical predictions $F_1(s)$ resp.\ $F_2(s)$. The choice of initial data and target set is specified at the end of Section \ref{sec:billiards}.
\label{f:f1f2}}
\end{figure}

The dynamics of a point particle in a billiard is integrable if there is a coordinate system in which the Hamilton-Jacobi equation separates.  All known examples in two dimensions involve either very particular polygonal billiards, whose dynamics unfolds to a linear flow on a torus, or billiards whose boundaries are aligned with elliptical coordinate lines (or the degenerate cases of circular or parabolic coordinates).  While many configurations can be constructed from arcs of confocal ellipses and hyperbolas, the most natural and studied is the ellipse billiard itself, of which the circle is a special case.  Scaling of escape from a circular billiard with a single small hole to a universal function of the product of hole size and time was observed in Fig.~3 of \cite{BD05}. We will here consider billiards in general ellipses, where the target set is a sub-interval of the boundary. Action-angle coordinates for the billiard flow have been described in the literature, for example in \cite{SMOD77}. For our purposes it will be simpler to formulate the dynamics in terms of the billiard map, which is the return map of the billiard flow to the boundary; see \cite{Tab05} for a detailed discussion. The billiard domain is confined by the ellipse
$\{ (b\cos\phi,a\sin\phi) \mid \phi\in[0,2\pi) \}$
with semi-axes $a\geq b$, eccentricity $e=\sqrt{1-b^2/a^2}$ and foci $(0,\pm ae)$.
The billiard dynamics conserves the kinetic energy $E=\frac12\|\vecxi\|^2$ (where $\vecxi$ denotes the particle's momentum) and the product $L_+L_-$ of angular momenta $L_\pm=x_1\xi_2-(x_2\mp ae)\xi_1$ about the foci. Note that a change in energy $E>0$ only affects the speed of the billiard particle but not its trajectory, and we will fix $E=\frac12$ in the following without loss of generality.

Each segment of the trajectory is tangent to a caustic given by a confocal conic of eccentricity
\begin{equation}\label{e:ve}
\ve=\sqrt{\frac{a^2e^2}{a^2e^2+L_+L_-}}\in(e,\infty)
\end{equation}
For $\ve<1$ we have elliptic caustics, where the orbit rotates around the foci. For $\ve=1$ we have the separatrix, where the orbit passes through the foci; this has zero probability with respect to an absolutely continuous distribution of initial conditions. For $\ve>1$ we have hyperbolic caustics, and the orbit passes between the foci.  Solving Eq.~(\ref{e:ve}) for $\vecxi$ gives two solutions, which for $\ve<1$ correspond to the direction of rotation of the orbit, and for $\ve>1$ are both contained in the closure of a single aperiodic orbit.

Following \cite{Tab05} in our notation, we parametrize the billiard boundary by the new parameter $\theta\in\TT$ defined by
\begin{equation}
\theta = \begin{cases}
\frac{F(\phi,\ve)}{F(2\pi,\ve)} \mod1 & (\ve<1)\\
\frac{F(\arcsin(\ve\sin\phi),\ve^{-1})}{F(2\pi,\ve^{-1})} \mod1 & (\ve>1),\end{cases} 
\end{equation}
where  $F$ is the elliptic integral of the first kind~\cite{OLBC10}
\begin{equation}
F(\phi,k)=\int_0^\phi\frac{d t}{\sqrt{1-k^2\sin^2t}}.
\end{equation}
The choice of branch for the $\arcsin$ (for $\ve>1$) depends on the choice of solution for $\vecxi$ in \eqref{e:ve}.
The billiard map reads in these new coordinates
\begin{equation}
\TT\to\TT,\qquad \theta\mapsto \theta+f(\ve) \mod1
\end{equation}
where
\begin{equation}
f(\ve)=\begin{cases}
\pm 2\frac{F\left(\arccos\sqrt{\frac{e^2(1-\ve^2)}{\ve^2(1-e^2)}},\ve\right)}{F(2\pi,\ve)}&(\ve<1)\\
2\frac{F\left(\arccos\sqrt{\frac{e^2(\ve^2-1)}{\ve^2-e^2}},\ve^{-1}\right)}{F(2\pi,\ve^{-1})}&(\ve>1).
\end{cases}
\end{equation}
Here, the $\pm$ (for $\ve<1$) again depends on the choice of solution for $\vecxi$ in \eqref{e:ve}.
The time between collisions with the boundary, averaged over the equilibrium measure associated with $\ve$, is given by
\begin{equation}
\bar{l}=\left\{\begin{array}{cc}\frac{2b\sqrt{1-e^2/\ve^2}\Pi(e^2,\ve)}{K(\ve)}&(\ve<1)\\
\frac{2b\sqrt{1-e^2/\ve^2}\Pi(e^2/\ve^2,\ve^{-1})}{K(\ve^{-1})}&(\ve>1)\end{array}\right.
\end{equation}
where $K(\ve)=F(\frac{\pi}{2},\ve)=\frac{1}{4}F(2\pi,\ve)$ and
\begin{equation}
\Pi(\alpha^2,k)=\int_0^{\frac{\pi}{2}}\frac{dt}{(1-\alpha^2\sin^2t) \sqrt{1-k^2\sin^2t}}
\end{equation}
are complete elliptic integrals of the first and third kind respectively~\cite{OLBC10}.  Even when $f(\ve)$ is rational,
hence the orbit is periodic (a set of zero measure of initial conditions), the mean collision time is independent
of the starting point, and hence given by the above formula~\cite{CCS93}.

We consider a single target set in the billiard's boundary given by the interval $\phi_0-\frac{\hole}{2}<\phi<\phi_0+\frac{\hole}{2}$.
If $\ve>1$, we assume the target intersects the region covered by the orbit, i.e., $\ve\sin\phi_0<1$.
In this case a single target in $\phi$ corresponds to two equal-sized targets in $\theta$ located at $\theta_0=\theta_0^{(1)}$ and $\theta_0^{(2)}$ (which are functions of $\phi_0$ and $\ve$). 
If $\ve< 1$, a single target in $\phi$ corresponds to a single target in $\theta$.  

For $\phi=\phi_0+s$ with $|s|$ small and $\theta$ (respectively $\theta_0$) the value defined by (5.2) for $\phi$ (respectively $\phi_0$),
\begin{equation}
\theta=\theta_0+s\left\{\begin{array}{cc}\frac{1}{F(2\pi,\varepsilon)\sqrt{1-\varepsilon^2\sin^2\phi_0}}&(\varepsilon<1)\\
\frac{\varepsilon}{F(2\pi,\varepsilon^{-1})\sqrt{1-\varepsilon^2\sin^2\phi_0}}&(\varepsilon>1)\end{array}\right\}+O(s^2).
\end{equation}
Up to a small error, which is negligible when $\rho\to0$, the target becomes the interval $\theta_0-\frac{\rho \ell}{2}<\theta<\theta_0+\frac{\rho \ell}{2}$ where
\begin{equation}
\ell=\ell(\varepsilon)=\left\{\begin{array}{cc}\frac{1}{F(2\pi,\varepsilon)\sqrt{1-\varepsilon^2\sin^2\phi_0}}&(\varepsilon<1)\\
\frac{\varepsilon}{F(2\pi,\varepsilon^{-1})\sqrt{1-\varepsilon^2\sin^2\phi_0}}&(\varepsilon>1)\end{array}\right.
.
\end{equation}

%
%

The circle is a special case, with $e=0$ and hence $\ve=0$.  The constant of motion is the angular momentum about the centre, $L=x_1\xi_2-x_2\xi_1$.  In this case
\begin{equation}
\theta=\frac{\phi}{2\pi},\quad f(0)=\pm\frac{1}{\pi}\arccos\frac{L}{a},\quad \ell=\frac{1}{2\pi},\quad \bar{l}=2\sqrt{a^2-L^2},
\end{equation}
which is consistent with the above expressions for ellipses in the limit $e\to 0$. For ellipses of small eccentricity, this approach gives a systematic expansion in powers of $e^2$.

Finally, we have for the mean return time \eqref{mean01}
\begin{equation}
\overline \sigma^{(k)} (\ve) = 
\begin{cases}
\frac{\bar{l}}{\ell(\ve)} & (\ve<1,\: \text{i.e.\ } k=1) \\[5pt]
\frac{\bar{l}}{2\ell(\ve)} & (\ve>1,\: \text{i.e.\ } k=2) .
\end{cases}
\end{equation}

For our numerical simulations of the first entry time, the relevant parameters used were as follows: $a=10$, $b\in\{6,8,10\}$ corresponding to $e\in\{0.8,0.6,0\}$
respectively.  The target was $2.8-5\times 10^{-5}<\phi<2.8+5\times 10^{-5}$, i.e.\ $\phi_0=2.8$ and $\rho=10^{-4}$.  The entry time distribution $\PP ( \widetilde \tau_1 > s )$ for the actual billiard flow was sampled by taking a {\em fixed} initial point $\vecx=(3,7)$ inside the ellipse, and $10^8$ initial directions $\vecxi\in\S^1$ chosen randomly with uniform angular distribution in the intervals $[2,2.6]$ or $[3.8,4.4]$ for the hyperbolic or elliptic
caustics, respectively. All the numerical curves are shown in Fig.~\ref{f:f1f2} and are identical within numerical errors too small to see on the plot; differences between the ellipse calculations and the theoretical predictions from Theorem \ref{thm:main000} are shown in the inset panels.

\section{Integrable flows in arbitrary dimension}\label{sec:general}

We now state the generalization of Theorem \ref{thm:main000} to arbitrary dimension $d\geq2$.
The basic setting is just as in Section \ref{sec:Integrable}, but with $\TT^d$ in place of $\TT^2$:
Let $\scrU$ be a bounded open subset of $\R^m$ for some $m\in\Z^+$, and let 
$\vecf:\scrU\to\R^d$ be a smooth function.
We consider the flow
\begin{equation}\label{eq:flowgen}
\varphi^t: \TT^d\times\scrU \to \TT^d\times\scrU, \quad (\vectheta,\vecJ) \mapsto  (\vectheta + t\, \vecf(\vecJ),\vecJ) .
\end{equation}
Let $\lambda$ be an absolutely continuous Borel probability measure on $\scrU$,
and let $\vectheta$ be a smooth map from $\scrU$ to $\TT^d$.
We will consider the random initial data $(\vectheta(\vecJ),\vecJ)$ in $\TT^d\times\scrU$,
where $\vecJ$ is a random point in $\scrU$ distributed according $\lambda$.

We next define the target sets.
Let us fix a map $\vecv\mapsto R_\vecv$, $\US\to\SO(d)$,
such that $R_\vecv \vecv=\vece_1$ for all $\vecv\in\US$,
and such that $\vecv\mapsto R_\vecv$ is smooth throughout $\US\setminus\{\vecv_0\}$,
where $\vecv_0$ is a fixed point in $\US$.
Fix $k\in\Z^+$ and for each $j=1,\ldots,k$, 
fix smooth functions $\vecu_j:\scrU\to\US$, $\vecphi_j:\scrU\to\TT^d$
and a bounded open subset $\Omega_j\subset\R^{d-1}\times\scrU$.
Set
\begin{align}
\scrD_\rho=\scrD_\rho^{(k)}:=\bigcup_{j=1}^k  \scrD_\rho(\vecu_j,\vecphi_j,\Omega_j) ,
\end{align}
where
\begin{align}
\scrD_\rho(\vecu,\vecphi,\Omega) := \biggl\{
\biggl(\vecphi(\vecJ)+\rho R_{\vecu(\vecJ)}^{-1}\cmatr{0}{\vecx},\,\vecJ\biggr)\in\TT^d\times\scrU
\:\bigg|\:(\vecx,\vecJ)\in\Omega_j\biggr\}.
\end{align}
Here we use the convention
\begin{equation}
\cmatr{0}{\vecx}:=\begin{pmatrix} 0 \\ x_1 \\ \vdots \\x_{d-1}\end{pmatrix} \in\R^d \quad\text{when}\quad 
\vecx=\begin{pmatrix}  x_1 \\ \vdots \\x_{d-1} \end{pmatrix} .
\end{equation}
Note that all points $R_{\vecu(\vecJ)}^{-1}\cmatr{0}{\vecx}$ lie in the linear subspace orthogonal to $\vecu(\vecJ)$ in $\R^d$.
We write $\Omega_j(\vecJ):=\{\vecx\in\R^{d-1}\col (\vecx,\vecJ)\in\Omega_j\}$,
and assume $\Omega_j(\vecJ)\neq\emptyset$ for all $\vecJ\in\scrU$.
As in Section \ref{sec:Integrable} we also impose the 
condition $\vecu_j(\vecJ)\cdot\vecf(\vecJ)>0$ for all $j\in\{1,\ldots,k\}$ and $\vecJ\in\scrU$,
which implies that each sub-target $\scrD_\rho(\vecu_j,\vecphi_j,\Omega_j)$  
is transversal to the flow direction.
Note that the target set $\scrD_\rho^{(k)}$ defined here generalizes the one introduced in 
Section \ref{sec:Integrable}.
Indeed, for $d=2$, and given smooth functions $\vecu_j:\scrU\to\S^1$, $\vecphi_j:\scrU\to\TT^2$, and 
$\ell_j:\scrU\to\RR_{>0}$ ($j=1,\ldots,k$),
we recover the target set in \eqref{target00} 
as $\bigcup_{j=1}^k  \scrD_\rho(\vecu_j,\vecphi_j,\Omega_j)$
where $\Omega_j=\{(s,\vecJ)\col\vecJ\in\scrU,\:-\frac12\ell_j(\vecJ)<s<\frac12\ell_j(\vecJ)\}$.

For any initial condition $(\vectheta,\vecJ)$,
let $\scrT(\vectheta,\vecJ,\scrD^{(k)}_\rho)$ be the set of hitting times, as in \eqref{scrTdef}.
This is a discrete subset of $\R_{>0}$, and we label its elements 
\begin{align}
0<t_1(\vectheta,\vecJ,\scrD^{(k)}_\rho)<t_2(\vectheta,\vecJ,\scrD^{(k)}_\rho)<\ldots .
\end{align}
Again by Santalo's formula, for any fixed $\vecJ\in\scrU$ such that the components of $\vecf(\vecJ)$ are not rationally related,
the first return time to $\scrD_\rho$ on the leaf $\TT^d\times\{\vecJ\}$
satisfies the formula
\begin{align}
\int_{\scrD_\rho}t_1(\vectheta,\vecJ,\scrD_\rho)\,d\nu_\vecJ(\vectheta)=1,
\end{align}
where $\nu_\vecJ$ is the invariant 
measure on $\scrD_\rho$ obtained by disintegrating Lebesgue measure on $\TT^d\times\{\vecJ\}$ 
with respect to the section $\scrD_\rho$ of the flow $\varphi^t$; explicitly
\begin{align}
\int_{\scrD_\rho} g\,d\nu_\vecJ=
\sum_{j=1}^k \bigl(\vecu_j(\vecJ)\cdot\vecf(\vecJ)\bigr) \int_{\rho\Omega_j(\vecJ)}
g \biggl(\vecphi_j(\vecJ)+R_{\vecu_j(\vecJ)}^{-1}\cmatr0{\vecx},\,\vecJ\biggr)\,d\vecx,
\qquad\forall g \in\C(\scrD_\rho).
\end{align}
It follows that the mean return time with respect to $\nu_\vecJ$ equals
\begin{align}\label{osigmakformula}
\frac{\osigma^{(k)}(\vecJ)}{\rho^{d-1}},\qquad\text{where }\hspace{10pt}
\overline \sigma^{(k)} (\vecJ):= \frac1{\sum_{j=1}^k\Leb(\Omega_j(\vecJ))\,\vecu_j(\vecJ)\cdot\vecf(\vecJ)},
\end{align}
with $\Leb$ denoting Lebesgue measure on $\R^{d-1}$.
If we also average over $\vecJ$ with respect to the measure $\lambda$
(assuming that the pushforward of $\lambda$ by $\vecf$ has no atoms at points with rationally related coordinates),
the mean return time becomes
\begin{align}
\frac{\osigma^{(k)}_\lambda}{\rho^{d-1}},\qquad\text{where }\hspace{10pt}
\overline \sigma^{(k)}_\lambda:=  \int_\scrU \overline \sigma^{(k)} (\vecJ) \lambda(d\vecJ).
\end{align}
As in Section \ref{sec:Integrable}, for $\vecJ$ a random point in $\scrU$ distributed according $\lambda$,
the hitting times $t_n(\vectheta(\vecJ),\vecJ,\scrD_\rho^{(k)})$ become random variables,
which we denote by $\tau_{n,\rho}^{(k)}$;
also $\osigma^{(k)}(\vecJ)$ becomes a random variable, which we denote by $\osigma^{(k)}$.
We say that $\lambda$ is {\em $\vecf$-regular} if
the pushforward of $\lambda$ under the map
\begin{equation}\label{f1defGEN}
\scrU \to\US, \qquad \vecJ \mapsto \frac{\vecf(\vecJ)}{\|\vecf(\vecJ)\|},
\end{equation}
is absolutely continuous with respect to Lebesgue measure on $\US$,
and we say the $k$-tuple of smooth functions $\vecphi_1,\ldots,\vecphi_k:\scrU \to\TT^d$ is 
\textit{$(\vectheta,\lambda)$-generic,} if
for all $\vecm=(m_1,\ldots,m_k)\in\ZZ^k\setminus\{\vecnull\}$ we have
\begin{equation}\label{thlmbgenDEF}
\lambda\bigg(\bigg\{ \vecJ\in\scrU : \sum_{j=1}^k m_j \, \big(\vecphi_j(\vecJ)
-\vectheta(\vecJ)\big) \in \RR \vecf(\vecJ) + \QQ^d \bigg\}\bigg) = 0.
\end{equation}

The following theorem generalizes Theorem \ref{thm:main000} to arbitrary dimension $d\geq2$.
\begin{thm}\label{thm:main001}
Let $\vecf:\scrU\to \RR^d$ and $\vectheta:\scrU\to\TT^d$ be smooth maps, 
$\lambda$ an absolutely continuous Borel probability measure on $\scrU$,
and for $j=1,\ldots,k$, let
$\vecu_j:\scrU \to\US$ and $\vecphi_j:\scrU \to\TT^d$ be smooth maps and 
$\Omega_j$ a bounded open subset of $\R^{d-1}\times\scrU$.
For each $j=1,\ldots,k$, assume that 
\begin{enumerate}[{\rm (i)}]
\item $\lambda(\vecu_j^{-1}(\{\vecv_0)\}))=0$ (where by assumption $\vecv_0$ is the point in $\US$ such that
$\vecv\mapsto R_\vecv$ is smooth throughout $\US\setminus\{\vecv_0\}$),
\item $\vecu_j(\vecJ)\cdot\vecf(\vecJ)>0$ for all $\vecJ\in\scrU$,
\item $\Omega_j$ has boundary of measure zero with respect to $\Leb\times\lambda$, 
\item $\Leb(\Omega_j(\vecJ))$ is a smooth and positive function of $\vecJ\in\scrU$.
\end{enumerate}
Also assume that $\lambda$ is $\vecf$-regular
and $(\vecphi_1,\ldots,\vecphi_k)$ is $(\vectheta,\lambda)$-generic.
Then there are sequences of random variables $(\tau_i)_{i=1}^\infty$ and $(\widetilde\tau_i)_{i=1}^\infty$ in $\RR_{>0}$ such that in the limit $\hole\to 0$, for every integer $N$,
\begin{equation}\label{thm:main001res1}
\bigg( \frac{\hole^{d-1} \tau_{1,\hole}^{(k)}}{\overline \sigma_\lambda^{(k)}},\ldots,\frac{\hole^{d-1} \tau_{N,\hole}^{(k)}}{\overline \sigma_\lambda^{(k)}} \bigg) \todist (\tau_1,\ldots,\tau_N),
\end{equation}
and
\begin{equation}\label{thm:main001res2}
\bigg( \frac{\hole^{d-1} \tau_{1,\hole}^{(k)}}{\overline \sigma^{(k)}},\ldots,\frac{\hole^{d-1} \tau_{N,\hole}^{(k)}}{\overline \sigma^{(k)}} \bigg) \todist (\widetilde\tau_1,\ldots,\widetilde\tau_N).
\end{equation}
\end{thm}

We next give an explicit description of the limit processes $(\tau_i)_{i=1}^\infty$ and $(\widetilde\tau_i)_{i=1}^\infty$
appearing in Theorem \ref{thm:main001}.
For a given affine Euclidean lattice $\scrL$ in $\R^d$ and a subset $\Omega\subset\R^{d-1}$, 
consider the cut-and-project set
\begin{align}\label{PLOMEGAdef}
\scrP(\scrL,\Omega):=\biggl\{t>0\col\cmatr{t}{\vecx}\in \scrL,\: \vecx\in-\Omega\biggr\}.
\end{align}
Fix an arbitrary (measurable) fundamental domain $F\subset\SL(d,\R)$ for 
$\SL(d,\R)/\SL(d,\Z)$,
and let $\mu_F$ be the (left and right) Haar measure on $\SL(d,\R)$ restricted to $F$,
normalized to be a probability measure.
If we choose $g\in F$ random according to $\mu_F$ then $g\Z^d$ represents a random Euclidean lattice in $\R^d$
(of covolume one).
Similarly, if $\vecalf$ is a random point in $\TT^d$, uniformly distributed and independent from $g$,
then the shifted lattice
$g(\Z^d+\vecalf)$ represents a random affine Euclidean lattice in $\R^d$.

Let us define
\begin{align}\label{vJdef}
\vecv(\vecJ)=\frac{\vecf(\vecJ)}{\|\vecf(\vecJ)\|}\in\US
\qquad (\vecJ\in\scrU).
\end{align}
For $j\in\{1,\ldots,k\}$ and $\vecJ\in\scrU$ we set
$\fR_j(\vecJ)=R_{\vecv(\vecJ)}R_{\vecu_j(\vecJ)}^{-1}\in\SO(d)$,
and let $\tR_j(\vecJ)$ be the bottom right $(d-1)\times(d-1)$ submatrix of $\fR_j(\vecJ)$.\label{tRjJdef}
In other words, $\tR_j(\vecJ)$ is the matrix of the linear map 
$\vecx\mapsto\biggl(\fR_j(\vecJ)\cmatr{0}{\vecx}\biggr)_{\!\!\perp}$ on $\R^{d-1}$,
where $\vecu_\perp:=(u_2,\ldots,u_d)\trans\in\R^{d-1}$ for $\vecu=(u_1,\ldots,u_d)\trans\in\R^d$.
Noticing that $\fR_j(\vecJ)$ is an orientation preserving isometry of $\R^d$ which takes
$\vece_1$ to $\fR_j(\vecJ)(\vece_1)$ and $\cmatr{0}{\R^{d-1}}$ onto $(\fR_j(\vecJ)(\vece_1))_\perp$,
we find that
\begin{align}\label{dettRjJ}
\det\tR_j(\vecJ)=\vece_1\cdot\fR_j(\vecJ)(\vece_1)
=\vece_1\cdot R_{\vecv(\vecJ)}(\vecu_j(\vecJ))=\vecu_j(\vecJ)\cdot\vecv(\vecJ)>0.
\end{align}
For $\vecJ\in\scrU$ we define
\begin{align}\label{oOmegadef1}
\oOmega_j(\vecJ):=\bigl(\osigma^{(k)}_\lambda\|\vecf(\vecJ)\|\bigr)^{1/(d-1)}\tR_j(\vecJ)\Omega_j(\vecJ)\subset\R^{d-1}
\end{align}
and
\begin{align}\label{tOmegadef2}
\tOmega_j(\vecJ):=\bigl(\osigma^{(k)}(\vecJ)\|\vecf(\vecJ)\|\bigr)^{1/(d-1)}\tR_j(\vecJ)\Omega_j(\vecJ)\subset\R^{d-1}.
\end{align}
Geometrically, thus, both $\oOmega_j(\vecJ)$ and $\tOmega_j(\vecJ)$ 
are obtained by orthogonally projecting the sub-target
$\{\vecx\in\TT^d\col(\vecx,\vecJ)\in\scrD_\rho(\vecu_j,\vecphi_j,\Omega_j)\}$
onto the hyperplane orthogonal to the flow direction $\vecf(\vecJ)$
(which is identified with $\R^{d-1}$ via the rotation $R_{\vecv(\vecJ)}$),
and then scaling the sets with appropriate scalar factors,
which in particular make $\oOmega_j(\vecJ)$ and $\tOmega_j(\vecJ)$ independent of $\rho$.

Now let $\vecJ$, $g$ and $\vecalf_1,\ldots,\vecalf_k$ be independent random points
in $\scrU$, $F$ and $\TT^d$, respectively, distributed according to 
$\lambda$, $\mu_F$ and $\Leb_{\TT^d}$.
We will prove in Section \ref{MAINPROOFsec} that the elements of the random set
\begin{align}\label{tauiGENDISCR1}
\bigcup_{j=1}^k\scrP(g(\Z^d+\vecalf_j),\oOmega_j(\vecJ)),
\end{align}
ordered by size, form precisely the sequence of random variables $(\tau_i)_{i=1}^\infty$ in Theorem \ref{thm:main001}.
Similarly the elements of
\begin{align}\label{ttauiGENDISCR1}
\bigcup_{j=1}^k\scrP(g(\Z^d+\vecalf_j),\tOmega_j(\vecJ)),
\end{align}
ordered by size, form the sequence of random variables $(\ttau_i)_{i=1}^\infty$.
We will also see in the proof that, for any $N\in\Z^+$,
both $(\tau_1,\ldots,\tau_N)$ and $(\widetilde\tau_1,\ldots,\widetilde\tau_N)$
have \textit{continuous} distributions, that is, the cumulative distribution functions
$\PP\bigl(\tau_n\leq T_n\text{ for }1\leq n\leq N)$ and
$\PP\bigl(\ttau_n\leq T_n\text{ for }1\leq n\leq N)$ depend
continuously on $(T_n)\in\R_{>0}^N$.

One verifies easily that the above description generalizes the one in Section \ref{sec:limit}.
Indeed, note that the image of the set $F$ in \eqref{F} under the map
\begin{align}
(x,y,\theta)\mapsto k_{\theta}\matr{\sqrt y}00{1/\sqrt y}\matr10x1
\end{align}
is a fundamental domain for $\SL(2,\R)/\SL(2,\Z)$,
and the pushforward of the measure $\mu_F$ in Section \ref{sec:limit} gives the measure $\mu_F$ 
considered in the present section.
Note also that for $d=2$, $\fR_j(\vecJ)$ is the $1\times1$ matrix with the single entry $\vecu_j(\vecJ)\cdot\vecv(\vecJ)$
(cf.\ \eqref{dettRjJ}), and now one checks that if
$\Omega_j=\{(s,\vecJ)\col\vecJ\in\scrU,\:-\frac12\ell_j(\vecJ)<s<\frac12\ell_j(\vecJ)\}$
then for any affine Euclidean lattice $\scrL$,
the cut-and-project set $\scrP(\scrL,\oOmega_j(\vecJ))$ equals $\scrP(\scrL,L_j(\vecJ))$,
and similarly $\scrP(\scrL,\tOmega_j(\vecJ))$ equals $\scrP(\scrL,\tL_j(\vecJ))$
(cf.\ \eqref{PLldef} and \eqref{PLOMEGAdef}).
%
%

Finally let us point out three invariance properties of the limit distributions.
First, both $(\tau_i)_{i=1}^\infty$ and $(\ttau_i)_{i=1}^\infty$ yield
\textit{stationary} point processes, i.e.\ the random set of time points $\{\tau_i\}$ has the same distribution as
$\{\tau_i-t\}\cap\R_{>0}$ for every fixed $t\geq0$,
and similarly for $\{\ttau_i\}$.
This is clear from the explicit description above, using in particular the fact that 
Lebesgue measure on the torus $\TT^d$ 
is invariant under any translation.
Secondly, by the same argument, the distributions of $(\tau_i)_{i=1}^\infty$ and $(\ttau_i)_{i=1}^\infty$
are not affected by any leaf-wise \textit{translation} of any of the sets $\Omega_j$,
i.e.\ replacing $\Omega_j$ by
the set $\{(\vecx+\vecg(\vecJ),\vecJ)\col(\vecx,\vecJ)\in\Omega_j\}$,
where $\vecg$ is any bounded continuous function from $\scrU$ to $\R^{d-1}$.
Thirdly, we point out the identity
\begin{align}
\scrP\biggl(\matr{{h}^{-1}}{\bn\trans}{\bn}{H}\scrL,H\Omega\biggr)={h}^{-1}\scrP(\scrL,\Omega),
\end{align}
which holds for any $\scrL$ and $\Omega$ as in \eqref{PLOMEGAdef},
and any $H\in\GL_{d-1}(\R)$ with $h=\det H>0$.
Note also that the map
\begin{align}
g\SL(d,\Z)\mapsto\matr{{h}^{-1}}{\bn\trans}{\bn}{H} g\SL(d,\Z)
\end{align}
is a measure preserving transformation of $\SL(d,\R)/\SL(d,\Z)$ onto itself.
For $d=2$ these two facts immediately lead to the formula \eqref{Fkinv} in Section \ref{sec:limit}.
For general $d\geq2$, the same facts imply for example that if $\vecu_1=\cdots=\vecu_k$
then the limit random sequences $(\tau_i)_{i=1}^\infty$ and $(\ttau_i)_{i=1}^\infty$ 
are not affected if 
$\Omega_j$ is replaced by $\{(H_1\vecx,\vecJ)\col(\vecx,\vecJ)\in\Omega_j\}$ 
simultaneously for all $j$,
where $H_1$ is any fixed $(d-1)\times(d-1)$ matrix with positive determinant.
Indeed, the given replacement has the effect that both $\osigma^{(k)}(\vecJ)$ and $\osigma^{(k)}_\lambda$ are multiplied by
the constant $(\det H_1)^{-1}$; thus both
$\oOmega_j(\vecJ)$ and $\tOmega_j(\vecJ)$ get transformed by the linear map
$H:=(\det H_1)^{-1/(d-1)}\tR_j(\vecJ)H_1\tR_j(\vecJ)^{-1}$,
which has determinant $1$ and is independent of $j$ since 
$\vecu_1(\vecJ)=\cdots=\vecu_k(\vecJ)$;
hence the statement follows from the two facts noted above. 

\section{An application of Ratner's Theorem}
\label{HOMDYNsec}

In this section we will introduce a homogeneous space $G/\Gamma$ which parametrizes such $k$-tuples of translates of
a common lattice as appear in \eqref{tauiGENDISCR1} and \eqref{ttauiGENDISCR1},
and then use Ratner's classification of unipotent-flow invariant measures   
to prove an asymptotic equidistribution result in $G/\Gamma$,
Theorem \ref{escMAINRATNERTHM1}, which will be a key ingredient for our proof of 
Theorem \ref{thm:main001} in Section \ref{MAINPROOFsec}.

Let $\SL(d,\R)$ act on $(\R^d)^k$ through
\begin{align}\label{SLdRaction}
M\vecv =M(\vecv_1,\ldots,\vecv_k)=(M\vecv_1,\ldots,M\vecv_k),
\end{align}
for $\vecv=(\vecv_1,\ldots,\vecv_k)\in(\R^d)^k$ and $M\in\SL(d,\R)$.
Let $G$ be the semidirect product
\begin{align*}
G=\SL(d,\R)\ltimes(\R^d)^k,
\end{align*}
with multiplication law
\begin{align*}
(M,\vecxi)(M',\vecxi')=(MM',\vecxi+M\vecxi').
\end{align*}
We extend the action of $\SL(d,\R)$ to an action of $G$ on $(\R^d)^k$, by defining
\begin{align}\label{GACTRd}
(M,\vecxi)\vecv:=M\vecv+\vecxi\qquad\quad \text{for}\quad
(M,\vecxi)\in G,\:\vecv\in(\R^d)^k.
\end{align}

Set $\Gamma=\SL(d,\Z)\ltimes(\Z^d)^k$ and $X=G/\Gamma$.
Let $\mu_X$ be the (left and right) Haar measure on $G$, normalized so as to induce a probability measure on $X$,
which we also denote by $\mu_X$.
We also set
\begin{align*}
D(\rho)=\text{diag}[\rho^{d-1},\rho^{-1},\ldots,\rho^{-1}]\in\SL(d,\R),\qquad\rho>0,
\end{align*}
and
\begin{align*}
n_-(\vecx)=\matr 1{\bn\trans}{\vecx}{1_{d-1}}\in\SL(d,\R),\qquad\vecx\in\R^{d-1}.
\end{align*}

We view $\SL(d,\R)$ as embedded in $G$ through $M\mapsto(M,\bn)$.
\begin{thm}\label{escMAINRATNERTHM1}
Let $M\in\SL(d,\R)$;
let $\scrU$ be an open subset of $\R^{d-1}$; 
let $\vecphi:\scrU\to(\R^d)^k$
be a Lipschitz map,
%
%
and let $\lambda$ be a Borel probability measure on $\scrU$ which is
absolutely continuous with respect to Lebesgue measure.
Writing $\vecphi(\vecv)=(\vecphi_1(\vecv),\ldots,\vecphi_k(\vecv))$, we assume that 
for every $\vecw=(w_1,\ldots,w_k)\in\Z^k\setminus\{\bn\}$,
\begin{align}\label{escMAINRATNERTHM1res1}
\lambda\biggl(\biggl\{\vecv\in\scrU\col
\sum_{j=1}^k w_j\cdot\vecphi_j(\vecv)\in\R M^{-1}\cmatr1{-\vecv}+\Q^d\biggr\}\biggr)=0.
\end{align}
Then for any bounded continuous function
$f:X\to\R$,
\begin{align}\label{escMAINRATNERTHM1res2}
\lim_{\rho\to0}\int_{\scrU}
f\bigl(D(\rho) n_-(\vecv) M (1_d,\vecphi(\vecv))\bigr)\,d\lambda(\vecv)
=\int_{X}f(g)\,d\mu_X(g).
\end{align}
\end{thm}
\begin{remark}
For related results on equidistribution of expanding translates of \textit{curves,}
cf.\ Shah, \cite[Thm.\ 1.2]{Shah2010}.
\end{remark}
\begin{remark}\label{EMconn}
The proof of Theorem \ref{escMAINRATNERTHM1}
extends trivially to the more general situation when $\Gamma$ is a subgroup of
$\SL(d,\Z)\ltimes(\Z^d)^k$ of finite index.
In this form, Theorem \ref{escMAINRATNERTHM1}
contains Elkies and McMullen, \cite[Thm.\ 2.2]{Elkies04} as a special case.
Indeed, applying Theorem \ref{escMAINRATNERTHM1}
with $d=2$, $k=1$, $M=\smatr0{-1}10$,
$\varphi(v)=\cmatr{x(v)+vy(v)}{y(v)}$ and 
$f(g):=f_0(M^{-1}g)$, where $f_0:X\to\R$ is an arbitrary bounded continuous function,
and noticing
$D(\rho) n_-(v)M(1_2,\varphi(v))=MD(\rho^{-1})\biggl(\matr1{-v}01,\cmatr{x(v)}{y(v)}\biggr)$,
we obtain
\begin{align*}
\lim_{s\to\infty}\int_{\scrU}
f_0\biggl(D(s) \biggl(\matr1{-v}01,\cmatr{x(v)}{y(v)}\biggr)\biggr)\,d\lambda(v)
=\int_{X}f_0(g)\,d\mu_X(g),
\end{align*}
provided that
\begin{align*}
\lambda\bigl(\bigl\{v\in\scrU\col x(v)\in\Q+\Q v\bigr\}\bigr)=0.
\end{align*}
Our proof of Theorem \ref{escMAINRATNERTHM1} follows the same basic strategy as the proof of
Thm.\ 2.2 in \cite{Elkies04}, but with several new complications arising. 
\end{remark}

\begin{remark}
Theorem \ref{escMAINRATNERTHM1} also generalizes
\cite[Thm. 5.2]{partI},
which is obtained by taking $k=1$ and $\vecphi(\vecv)=\vecphi$ a constant vector independent of $\vecv$.
Indeed note that \eqref{escMAINRATNERTHM1res2} in this case is equivalent with $\vecphi\notin\Q^d$.
(To translate into the setting of \cite{partI}, 
where vectors are represented as row matrices and one considers $\Gamma\bs G$ in place of $G/\Gamma$;
apply the map $(M,\vecxi)\mapsto (M\trans,\vecxi\trans)$.)
\end{remark}

We now give the proof of Theorem \ref{escMAINRATNERTHM1}; it extends until page \pageref{escMAINRATNERTHM1pfFINISH}.
Let $M,\scrU,\vecphi,\lambda$ satisfy all the assumptions of Theorem \ref{escMAINRATNERTHM1}.
As an initial reduction, let us note that by a standard approximation argument
where one removes from $\scrU$ a subset of small $\lambda$-measure,
we may in fact assume that $\scrU$ is bounded, 
and furthermore that there is a constant $B>0$ such that $\lambda(A)\leq B\Leb(A)$ for every Borel set $A\subset\scrU$.
\label{INITIALREDUCTION}
(We will only use these properties in the proof of Lemma \ref{MUTOT0lem} below.)

For each $\rho>0$, let $\mu_\rho$ be the probability measure on $X$ defined by
\begin{align}
\mu_\rho(f)=\int_{\scrU} f\bigl(D(\rho)n_-(\vecv)M(1_d,\vecphi(\vecv))\bigr)\,d\lambda(\vecv),
\qquad f\in \C_c(X).
\end{align}
Our task is to prove that $\mu_\rho$ converges weakly to $\mu_X$ as $\rho\to0$.
In fact it suffices to prove that $\mu_\rho(f)\to\mu_X(f)$ holds for every
function $f$ in the space of continuous compactly supported functions on $X$, $\C_c(X)$.
Recall that the unit ball in the dual space of $\C_c(X)$ is compact in the weak* topology
(Alaoglu's Theorem).
Hence by a standard subsequence argument, it suffices to prove that 
every weak* limit of $(\mu_\rho)$ as $\rho\to0$ must equal $\mu_X$.
Thus from now on, we let $\mu$ be a weak* limit of $(\mu_\rho)$,
i.e.\ $\mu$ is a Borel measure (apriori not necessarily a probability measure) on
$X$,
and we have $\mu_{\rho_j}(f)\to\mu(f)$ for every $f\in\C_c(X)$,
where $(\rho_j)$ is a fixed sequence of positive numbers tending to $0$.
Our task is to prove $\mu=\mu_X$.

Let $\pi:G\to\SL(d,\R)$ be the projection $(M,\vecxi)\mapsto M$;
this map induces a projection $X\to X':=\SL(d,\R)/\SL(d,\Z)$ which we also call $\pi$.
Let $\mu_{X'}$ be the unique $\SL(d,\R)$ invariant probability measure on $X'$.

\begin{lem}\label{MUPROJMUXPlem}
$\pi_*\mu=\mu_{X'}$.
\end{lem}
\begin{proof}
For any $f\in\C_c(X')$ we have
\begin{align}\label{MUPROJMUXPlempf1}
\pi_*\mu(f)=\lim_{j\to\infty}\mu_{\rho_j}(f\circ\pi)=
\lim_{j\to\infty}\int_\scrU f\bigl(D(\rho_j)n_-(\vecv)M\bigr)\,d\lambda(\vecv)
=\mu_{X'}(f).
\end{align}
For the last equality, cf., e.g., \cite[Prop.\ 2.2.1]{KleinbockMargulis96}.
(The point here is that $f$ is averaged along expanding translates of a horospherical subgroup,
and such translates can be proved to become asymptotically equidistributed
using the so called thickening method,
originally introduced in the 1970 thesis of Margulis \cite{Margulisthesis}.)
\end{proof}

\begin{lem}\label{escMAINRATNERTHM1pflem1}
$\mu$ is invariant under $n_-(\vecx)$ for every $\vecx\in\R^{d-1}$.
\end{lem}
\begin{proof}
(Cf.\ \cite[Thm.\ 2.5]{Elkies04}.)
Let $\lambda'\in\L^1(\R^{d-1})$ be the Radon-Nikodym derivative of $\lambda$ with respect to Lebesgue measure
(thus $\lambda'(\vecv)=0$ for $\vecv\notin\scrU$).
Let $f\in \C_c(X)$ and $\vecx\in\R^{d-1}$ be given,
and define $f_1\in\C_c(X)$ through $f_1(p)=f(n_-(\vecx)p)$.
Then our task is to prove that $\mu(f_1)=\mu(f)$,
viz., to prove that the difference
\begin{align*}
\int_{\scrU} f\bigl(n_-(\vecx)D(\rho_j)n_-(\vecv)M(1_d,\vecphi(\vecv))\bigr)\lambda'(\vecv)\,d\vecv
-\int_{\scrU} f\Bigl(D(\rho_j)n_-(\vecw)M(1_d,\vecphi(\vecw))\Bigr)\lambda'(\vecw)\,d\vecw
\end{align*}
tends to $0$ as $j\to\infty$.
Using $n_-(\vecx)D(\rho_j)=D(\rho_j)n_-(\rho_j^d\vecx)$ 
and substituting $\vecv=\vecw-\rho_j^d\vecx$ in the first integral,
the difference can be rewritten as
\begin{align}\notag
\int_{(\scrU+\rho_j^d\vecx)\cap\scrU}
\Bigl(f\bigl(D(\rho_j)n_-(\vecw)M(1_d,\vecphi(\vecw-\rho_j^d\vecx))\bigr)
-f\bigl(D(\rho_j)n_-(\vecw)M(1_d,\vecphi(\vecw))\bigr)\Bigr)\lambda'(\vecw)\,d\vecw
\\
+\int_{\scrU+\rho_j^d\vecx}
f\bigl(D(\rho_j)n_-(\vecw)M(1_d,\vecphi(\vecw-\rho_j^d\vecx))\bigr)
\bigl(\lambda'(\vecw-\rho_j^d\vecx)-\lambda'(\vecw)\bigr)\,d\vecw
\\\notag
-\int_{\scrU\setminus(\scrU+\rho_j^d\vecx)}
f\bigl(D(\rho_j)n_-(\vecw)M(1_d,\vecphi(\vecw))\bigr)\lambda'(\vecw)\,d\vecw.
\end{align}
The absolute value of this expression is bounded above by
\begin{align}\label{escMAINRATNERTHM1pflem1pf1}
&\sup_{\vecw\in(\scrU+\rho_j^d\vecx)\cap\scrU}\:
\Bigl|f\bigl(D(\rho_j)n_-(\vecw)M(1_d,\vecphi(\vecw-\rho_j^d\vecx))\bigr)
-f\bigl(D(\rho_j)n_-(\vecw)M(1_d,\vecphi(\vecw))\bigr)\Bigr|
\\\notag
&\hspace{200pt}
+\Bigl(\sup_{X}|f|\Bigr)\int_{\R^{d-1}}\bigl|\lambda'(\vecw-\rho_j^d\vecx)-\lambda'(\vecw)\bigr|\,d\vecw.
\end{align}
By assumption, there exists $C>0$ such that
$\|\vecphi(\vecw')-\vecphi(\vecw)\|\leq C\|\vecw'-\vecw\|$ for all $\vecw,\vecw'\in\scrU$,
where in the left hand side
$\|\cdot\|$ is the standard Euclidean norm on $(\R^d)^k$.
In particular for any $\vecw\in(\scrU+\rho_j^d\vecx)\cap\scrU$ 
we have $\vecphi(\vecw-\rho_j^d\vecx)=\vecphi(\vecw)+\veceta$
for some $\veceta=\veceta(\vecw,j)\in(\R^d)^k$ satisfying $\|\veceta\|\leq C\rho_j^{d}\|\vecx\|$,
and thus 
\begin{align*}
D(\rho_j)n_-(\vecw)M(1_d,\vecphi(\vecw-\rho_j^d\vecx))
&=D(\rho_j)n_-(\vecw)M(1_d,\vecphi(\vecw)+\veceta)
\\
&=\bigl(1_d,D(\rho_j)n_-(\vecw)M\veceta\bigr)\,D(\rho_j)\,n_-(\vecw)\,M\,\bigl(1_d,\vecphi(\vecw)\bigr).
\end{align*}
Now if $M\veceta=(\veceta_1',\ldots,\veceta_k')$ and 
$\veceta_\ell'=(\eta_{\ell,1}',\ldots,\eta_{\ell,d}')\trans$ for each $\ell$, then the $\ell$th component of
$D(\rho_j)n_-(\vecw)M\veceta$
equals
$\rho_j^{-1}\eta'_{\ell,1}\cmatr{\rho_j^{d}}{\vecw}+\rho_j^{-1}(0,\eta'_{\ell,2},\cdots,\eta'_{\ell,d})\trans$.
Now $\|M\veceta\|\ll_{C,M}\rho_j^d\|\vecx\|$,
and thus
the element $(1_d,D(\rho_j)n_-(\vecw)M\veceta)$ tends to the identity in $G$ as $j\to\infty$,
uniformly over all $\vecw\in(\scrU+\rho_j^d\vecx)\cap\scrU$.
But $f$ is uniformly continuous since $f\in \C_c(X)$; 
hence it follows that the first term in the right hand side of \eqref{escMAINRATNERTHM1pflem1pf1} 
tends to zero as $j\to\infty$.
Also the second term 
tends to zero; cf., e.g., \cite[Prop.\ 8.5]{Folland}.
This completes the proof of the lemma.
\end{proof}

Since $\mu$ is $n_-(\R^{d-1})$-invariant,   
we can apply ergodic decomposition to $\mu$:
Let $\scrE$ be the set of ergodic $n_-(\R^{d-1})$-invariant probability measures on $X$,
provided with its usual Borel $\sigma$-algebra;
then there exists a unique Borel probability measure $P$ on $\scrE$ such that
\begin{align}\label{ERGDEC}
\mu=\int_{\scrE}\nu\,dP(\nu).
\end{align}
Cf., e.g., \cite[Thm.\ 4.4]{Varadarajan}.
Note that \eqref{ERGDEC} together with Lemma \ref{MUPROJMUXPlem} implies 
$\mu_{X'}=\pi_*\mu=\int_{\scrE}\pi_*\nu\,dP(\nu)$,
and for each $\nu\in\scrE$, $\pi_*\nu$ is an ergodic $n_-(\R^{d-1})$-invariant measure on $X'$.
Hence in fact $\pi_*\nu=\mu_{X'}$ for $P$-almost all $\nu\in\scrE$,
by uniqueness of the ergodic decomposition of $\mu_{X'}$.

Now fix an arbitrary $\nu\in\scrE$ satisfying $\pi_*\nu=\mu_{X'}$.
We now apply Ratner's classification of unipotent-flow invariant measures, \cite[Thm 3]{mR91a}, to $\nu$.
Let $H$ be the closed (Lie) subgroup of $G$ given by
\begin{align*}
H=\{g\in G\col g_*\nu=\nu\},
\end{align*}
where $g_*\nu$ denotes the push-forward of $\nu$ by the map $x\mapsto gx$ on $X$
(viz., $(g_*\nu)(B):=\nu(g^{-1}B)$ for any Borel set $B\subset X$).
Note that
\begin{align}\label{Hcontainsnm}
n_-(\R^{d-1})\subset H,
\end{align}
by definition.
The conclusion from \cite[Thm 3]{mR91a} is that there is some $g_0\in G$
such that $\nu(Hg_0\Gamma/\Gamma)=1$.
Note that in this situation the measure $\nu_0:=g_{0*}^{-1}\nu$ is
$g_0^{-1}Hg_0$ invariant 
and $\nu_0(g_0^{-1}Hg_0\Gamma/\Gamma)=1$. 
Hence under the standard identification of
$g_0^{-1}Hg_0\Gamma/\Gamma$ with the homogeneous space
$g_0^{-1}Hg_0/(\Gamma\cap g_0^{-1}Hg_0)$
(viz., $h\Gamma\mapsto h(\Gamma\cap g_0^{-1}Hg_0)$ for $h\in g_0^{-1}Hg_0$),
$\nu_0$ is the unique invariant probability measure on
$g_0^{-1}Hg_0/(\Gamma\cap g_0^{-1}Hg_0)$,
induced from a Haar measure on $g_0^{-1}Hg_0$.
In particular $\Gamma\cap g_0^{-1}Hg_0$ is a lattice in $g_0^{-1}Hg_0$,
and both $g_0^{-1}Hg_0\Gamma/\Gamma$ and $Hg_0\Gamma/\Gamma$ are closed subsets of $X$
(cf.\ also \cite[Thm.\ 1.13]{mR72});
furthermore $\supp(\nu)=Hg_0\Gamma/\Gamma$.
\begin{lem}\label{piHFULLlem}
In this situation, $\pi(H)=\SL(d,\R)$.
\end{lem}
\begin{proof}
(Cf.\ \cite[Thm.\ 2.8]{Elkies04}.)
We have $\pi(\supp\nu)=\supp\pi_*\nu$, since $\pi:X\to X'$ has compact fibers,
and $\supp\pi_*\nu=X'$, since we are assuming $\pi_*\nu=\mu_{X'}$.
Also $\supp\nu=Hg_0\Gamma/\Gamma$.
Hence 
$\pi(H)\pi(g_0)\SL(d,\Z)=\SL(d,\R)$,
and thus $\pi(H)=\SL(d,\R)$.
\end{proof}
In the next lemma we deduce from \eqref{Hcontainsnm} and Lemma \ref{piHFULLlem}
an explicit presentation of $H$.
For $\vecxi=(\vecxi_1,\ldots,\vecxi_k)\in(\R^d)^k$ and $\vecu=(u_1,\ldots,u_k)\in\R^k$, let us introduce the notation
\begin{align*}
\vecxi\cdot\vecu:=\sum_{j=1}^k u_j\vecxi_j\in\R^d.
\end{align*}
Given any linear subspace $U\subset\R^k$, we let $L(U)$ be the linear subspace consisting of all
$\vecxi\in(\R^d)^k$ satisfying $\vecxi\cdot\vecu=\bn$ for all $\vecu\in U^\perp$,
where $U^\perp$ is the orthogonal complement of $U$ in $\R^k$ with respect to the standard inner product.
(It is natural to identify $\vecxi=(\vecxi_1,\ldots,\vecxi_k)$ with the $d\times k$-matrix with columns
$\vecxi_1,\ldots,\vecxi_k$; then $\vecxi\cdot\vecu$ is simply matrix multiplication,
and $L(U)$ is the space of all $d\times k$-matrices such that every \textit{row} vector is in $U$.)
Note that $L(U)$ is closed under multiplication from the left by any $\SL(d,\R)$-matrix.
Hence the following is a closed Lie subgroup of $G$:
\begin{align*}
H_U:=\SL(d,\R)\ltimes L(U).
\end{align*}
Let $\vece_1=(1,0,\ldots,0)\trans\in\R^d$. Then $\vece_1^\perp=\{(0,\xi_2,\ldots,\xi_d)\trans\col \xi_j\in\R\}$,
and $(\vece_1^\perp)^k$ is a linear subspace of $(\R^d)^k$.
\begin{lem}\label{HEXPLlem}
There exist 
$U\subset\R^k$ and $\vecxi\in(\vece_1^\perp)^k$ such that
$H=(1_d,\vecxi)H_U(1_d,\vecxi)^{-1}$.
\end{lem}
\begin{proof}
Set $V=\{\vecxi\in(\R^d)^k\col (1_d,\vecxi)\in H\}$;
this is a closed subgroup of $\langle(\R^d)^k,+\rangle$,
and it follows using Lemma \ref{piHFULLlem} that $V$ is $\SL(d,\R)$-invariant,
i.e.\ $M\vecxi\in V$ whenever $M\in\SL(d,\R)$ and $\vecxi\in V$.
Let $\lsl(d,\R)$ be the Lie algebra of $\SL(d,\R)$,
i.e.\ the Lie algebra of $d\times d$ matrices with trace $0$.
Then for every ${\vecxi}\in V$, $A\in\lsl(d,\R)$ and $n\in\Z^+$ we have
$n(\exp(n^{-1}A){\vecxi}-{\vecxi})\in V$, and since $V$ is closed,
letting $n\to\infty$ we obtain $A{\vecxi}\in V$.
Using the formula $E_{ij}E_{ji}=E_{ii}$,
where $E_{ij}$ denotes the $d\times d$ matrix which has $(i,j)$th entry 1 and all other entries 0,
the last invariance is upgraded to:
$A{\vecxi}\in V$ for \textit{any} real $d\times d$-matrix $A$ and $\vecxi\in V$.
This is easily seen to imply 
$V=L(U)$ for some subspace $U\subset\R^k$.
Thus
\begin{align*}
N=H\cap\pi^{-1}(\{1_d\})=\{1_d\}\ltimes L(U).
\end{align*}
This is a normal subgroup of $G$.
Given any $M\in\SL(d,\R)$, by Lemma~\ref{piHFULLlem} there exists some $\vecxi\in(\R^d)^k$ such that 
$h:=(M,\vecxi)\in H$, and then $H\cap\pi^{-1}(\{M\})=Nh$.
Using also the fact that $(\R^d)^k=L(U)\oplus L(U^\perp)$ it follows that for each $M\in\SL(d,\R)$ there is a 
\textit{unique} $\veceta\in L(U^\perp)$ such that $(M,\veceta)\in H$.
Hence if we let $H'$ be the closed Lie subgroup of $H_{U^\perp}$ given by 
\begin{align*}
H':=H\cap H_{U^\perp},
\end{align*}
then $H'$ contains exactly one element above each $M\in\SL(d,\R)$,
and 
$H=NH'=H'N$.
Note that the unipotent radical of $H_{U^\perp}=\SL(d,\R)\ltimes L(U^\perp)$ equals $\{1_d\}\ltimes L(U^\perp)$,
and thus $H'$ is a Levi subgroup of $H_{U^\perp}$.
Hence by Malcev's Theorem
(\cite{Malcev42}; \cite[Ch.\ III.9]{Jacobson62})
there exists some $\vecxi\in L(U^\perp)$ such that 
$H'=(1_d,\vecxi)\SL(d,\R)(1_d,\vecxi)^{-1}$.
(Recall that we view $\SL(d,\R)$ as embedded in $G$ through $M\mapsto(M,\bn)$.)
Hence
\begin{align*}
H=NH'=
(1_d,\vecxi)H_U(1_d,\vecxi)^{-1}.
\end{align*}
Finally using \eqref{Hcontainsnm} we see that $\vecxi$ must lie in $(\vece_1^\perp)^k$.
\end{proof}


Next, for any linear subspace $U\subset\R^k$, $q\in\Z^+$ and $\vecxi\in(\vece_1^\perp)^k$,
we set
\begin{align}
\scrX_{U,q,\vecxi}=\{g\Gamma\col g\in G,\: g^{-1}\vecxi\in L(U)+q^{-1}(\Z^d)^k\}\subset X.
\end{align}
Note here that the set $L(U)+q^{-1}(\Z^d)^k$ is invariant under the action of $\Gamma$;
hence if $g^{-1}\vecxi\in L(U)+q^{-1}(\Z^d)^k$ then also 
$(g\gamma)^{-1}\vecxi\in L(U)+q^{-1}(\Z^d)^k$ for every $\gamma\in\Gamma$.
Note also that if $U$ intersects $\Z^k$ in a lattice
(viz., $\Z^k\cap U$ contains an $\R$-linear basis for $U$),
then $L(U)+q^{-1}(\Z^d)^k$ is a closed subset of
$(\R^d)^k$, and it follows that $\scrX_{U,q,\vecxi}$ is a closed subset of $X$.

\begin{lem}\label{SUPPORTlem}
There exist $q\in\Z^+$ and $\vecxi\in(\vece_1^\perp)^k$,
and a linear subspace $U\subset\R^k$ which intersects $\Z^k$ in a lattice, such that
$\supp(\nu)=Hg_0\Gamma/\Gamma\subset\scrX_{U,q,\vecxi}$.
\end{lem}
\begin{proof}
Take $U\subset\R^k$ and $\vecxi\in(\vece_1^\perp)^k$ as in Lemma \ref{HEXPLlem};
then $H=(1_d,\vecxi)H_U(1_d,-\vecxi)$.
Now $\Gamma$ intersects $g_0^{-1}Hg_0$ in a lattice;
hence if $g=g_0^{-1}(1_d,\vecxi)$ then $g^{-1}\Gamma g$ intersects $H_U$ in a lattice.
Set $\vecxi'=g_0^{-1}\vecxi$; then $g=(M,\vecxi')=(1_d,\vecxi')M$ for some $M\in\SL(d,\R)$,
and since $M$ normalizes $H_U$, it follows that 
$\widetilde\Gamma:=(1_d,\vecxi')^{-1}\Gamma(1_d,\vecxi')\cap H_U$ is a lattice in $H_U$.
By \cite[Cor.\ 8.28]{mR72},
this implies that 
$\widetilde\Gamma_r:=\{\vecv\in L(U)\col (1_d,\vecv)\in\widetilde\Gamma\}=(\Z^d)^k\cap L(U)$ is a lattice in $L(U)$,
and $\pi(\widetilde\Gamma)$ is a lattice in $\SL(d,\R)$.
The first condition implies that $\Z^k\cap U$ contains an $\R$-linear basis for $U$,
i.e.\ $U$ intersects $\Z^k$ in a lattice.
Next we compute
\begin{align*}
\pi(\widetilde\Gamma)
=\{\gamma\in\SL(d,\Z)\col (1_d-\gamma)\vecxi'\in L(U)+(\Z^d)^k\}.
\end{align*}
This is a subgroup of $\SL(d,\Z)$ and a lattice in $\SL(d,\R)$;
hence $\pi(\widetilde\Gamma)$ must be a subgroup of finite index in $\SL(d,\Z)$.
Now fix any $\gamma\in\pi(\widetilde\Gamma)$ for which $1_d-\gamma$ is invertible 
(for example we can take $\gamma$ as an appropriate integer power of any given hyperbolic element in $\SL(d,\Z)$).
Then $1_d-\gamma\in\GL(d,\Q)$,
and we conclude $\vecxi'\in(1_d-\gamma)^{-1}(L(U)+(\Z^d)^k)\subset L(U)+(\Q^d)^k$,
i.e.\ $\vecxi'=\vecu+q^{-1}\vecm$ for some $\vecu\in L(U)$, $q\in\Z_{>0}$ and $\vecm\in(\Z^d)^k$.

Now for every $g\in Hg_0\Gamma$ we have
%
%
$(1_d,-\vecxi)g_0\Gamma g^{-1}(1_d,\vecxi)\cap H_U\neq\emptyset$,
i.e.\ there is some $\gamma\in\Gamma$ such that 
$(1_d,-\vecxi)g_0\gamma g^{-1}(1_d,\vecxi)\bn\in L(U)$,
or equivalently $g^{-1}\vecxi\in\gamma^{-1} g_0^{-1}(1_d,\vecxi)L(U)$.
But we have $g_0^{-1}(1_d,\vecxi)=(M,\vecxi')=(M,\vecu+q^{-1}\vecm)$
and hence 
$\gamma^{-1} g_0^{-1}(1_d,\vecxi)L(U)=\gamma^{-1}(L(U)+q^{-1}\vecm)\subset L(U)+q^{-1}(\Z^d)^k$.
Hence every $g\in Hg_0\Gamma$ satisfies $g^{-1}\vecxi\in L(U)+q^{-1}(\Z^d)^k$,
i.e.\ we have $Hg_0\Gamma/\Gamma\subset\scrX_{U,q,\vecxi}$.
\end{proof}


Recall that we have fixed $\mu$ as an arbitrary weak* limit of $(\mu_\rho)$ as $\rho\to0$.
The proof of the following Lemma \ref{MUTOT0lem} makes crucial use of the genericity assumption
\eqref{escMAINRATNERTHM1res1} in Theorem \ref{escMAINRATNERTHM1};
later Lemma \ref{MUTOT0lem} combined with Lemma \ref{SUPPORTlem}
will allow us to conclude that in the ergodic decomposition \eqref{ERGDEC},
we must have $\nu=\mu_X$ for $P$-almost all $\nu$.

\begin{lem}\label{MUTOT0lem}
Let $q\in\Z^+$ and let $U$ be a linear subspace of $\R^k$ of dimension $<k$ which intersects $\Z^k$ in a lattice.
Then 
$\mu\bigl(\cup_{\vecxi\in(\vece_1^\perp)^k}\:\scrX_{U,q,\vecxi}\bigr)=0$.
\end{lem}
\begin{proof}
Let $\scrB_C^d$ be the closed ball of radius $C$ in $\R^d$ centered at the origin.
It suffices to prove that for each $C>0$, the set
\begin{align}
\scrX_{U,q,C}:=\bigcup_{\vecxi\in(\scrB_C^d\cap\vece_1^\perp)^k}\scrX_{U,q,\vecxi}\subset X
\end{align}
satisfies $\mu\bigl(\scrX_{U,q,C}\bigr)=0$.
Let $\scrN$ be the family of open subsets of $G$ containing the identity element.
Then for any $\Omega\in\scrN$, $\Omega\scrX_{U,q,C}$ is an open set in $X$ containing $\scrX_{U,q,C}$.
Hence, since $\mu$ is a weak* limit of $(\mu_\rho)$ as $\rho\to0$ along some subsequence, 
it now suffices to prove that for every $\ve>0$ there exists some $\Omega\in\scrN$ such that
$\limsup_{\rho\to0}\mu_\rho(\Omega\scrX_{U,q,C})<\ve$.
We have $g\Gamma\in\scrX_{U,q,C}$ if and only if
the set $g(L(U)+q^{-1}(\Z^d)^k)$ in $(\R^d)^k$ has some point in common with $(\scrB_C^d\cap\vece_1^\perp)^k$.
The latter is a compact set, which for any $\eta>0$ is contained in the open set $V_\eta^k$, 
where (after increasing $C$ by $1$)
\begin{align}
V_\eta:=\bigl\{(\xi_1,\ldots,\xi_d)\trans\col |\xi_1|<\eta,\: \|(\xi_2,\ldots,\xi_d)\|<C\bigr\}\subset\R^d.
\end{align}
Hence for every $\eta>0$, there exists some $\Omega\in\scrN$ such that
\begin{align}
\Omega\scrX_{U,q,C}\subset\scrX_{U,q,C,\eta}:=\bigl\{g\Gamma\col g(L(U)+q^{-1}(\Z^d)^k)\cap V_\eta^k\neq\emptyset\bigr\}.
\end{align}
Hence it now suffices to prove
\begin{align}
\lim_{\eta\to0}\limsup_{\rho\to0}\mu_\rho(\scrX_{U,q,C,\eta})=0.
\end{align}

By the definition of $\mu_\rho$
we have $\mu_\rho(\scrX_{U,q,C,\eta})=\lambda(T_\rho)$,
where
\begin{align*}
T_\rho &=\bigl\{\vecv\in\scrU\col D(\rho)n_-(\vecv)M(1_d,\vecphi(\vecv))\in \scrX_{U,q,C,\eta}\bigr\}
\\
&=\bigl\{\vecv\in\scrU\col D(\rho)n_-(\vecv)M(L(U)+q^{-1}(\Z^d)^k+\vecphi(\vecv))\cap V_\eta^k\neq\emptyset\bigr\}.
\end{align*}
It follows from our assumptions on $U$ that there exists some $\vecw\in\Z^k\setminus\{\bn\}$
such that $U$ is contained in $\vecw^\perp$, the orthogonal complement of $\vecw$ in $\R^k$.
Now every $\vecxi\in L(U)+q^{-1}(\Z^d)^k$
satisfies $\vecxi\cdot\vecw\in q^{-1}\Z^d$,
and hence for any $\vecv\in\scrU$, every $\vecxi$ in the set $D(\rho)n_-(\vecv)M(L(U)+q^{-1}(\Z^d)^k+\vecphi(\vecv))$
satisfies
\begin{align}
\vecxi\cdot\vecw\in D(\rho)n_-(\vecv)M(q^{-1}\Z^d+\vecphi(\vecv)\cdot\vecw).
\end{align}
But on the other hand, for every $\vecxi\in V_\eta^k$ we have 
\begin{align}
\vecxi\cdot\vecw\in \|\vecw\|V_\eta=
\bigl\{(\xi_1,\ldots,\xi_d)\trans\col |\xi_1|<\|\vecw\|\eta,\: \|(\xi_2,\ldots,\xi_d)\|<\|\vecw\|C\bigr\}.
\end{align}
Hence
\begin{align}
T_\rho\subset\bigl\{\vecv\in\scrU\col D(\rho)n_-(\vecv)M(q^{-1}\Z^d+\vecphi(\vecv)\cdot\vecw)\cap 
\|\vecw\|V_\eta\neq\emptyset\bigr\}.
\end{align}
Therefore, if we alter the constant ``$C$'' appropriately in the definition of $V_\eta$,
we see that it now suffices to prove that
\begin{align}\label{MUTOT0lempf2}
\lim_{\eta\to0}\limsup_{\rho\to0}\lambda\biggl(\bigcup_{\vecm\in q^{-1}\Z^d}\tT_\rho^\vecm\biggr)=0,
\end{align}
where
\begin{align}
\tT_\rho^\vecm :&=\bigl\{\vecv\in\scrU\col D(\rho) n_-(\vecv)M(\vecm+\vecphi(\vecv)\cdot\vecw)\in V_\eta\bigr\}.
\end{align}

For $\vecv\in\R^{d-1}$ and $\vecu=(u_1,\ldots,u_d)\trans\in\R^d$ let us write
$\vecu_\perp:=(u_2,\ldots,u_d)\trans\in\R^{d-1}$ and
$\vecell_\vecv(\vecu)=u_1\vecv+\vecu_\perp\in\R^{d-1}$,
so that $n_-(\vecv)\vecu =\cmatr{\vece_1\cdot\vecu}{\vecell_\vecv(\vecu)}$.
Then the set $\tT_\rho^\vecm$ can be expressed as
\begin{align}
\tT_\rho^\vecm=X_\rho^\vecm\cap Y_\rho^\vecm,
\end{align}
where
\begin{align*}
X_\rho^{\vecm}=\Bigl\{\vecv\in \scrU\col
\vecell_\vecv(M(\vecm+\vecphi(\vecv)\cdot\vecw))\in\scrB_{C\rho}^{d-1}\Bigr\}
\end{align*}
and
\begin{align*}
Y_\rho^\vecm=\Bigl\{\vecv\in\scrU\col\vece_1\cdot
M(\vecm+\vecphi(\vecv)\cdot\vecw)
\in(-\eta\rho^{1-d},\eta\rho^{1-d})\Bigr\}.
\end{align*}
Let us note 
that the genericity assumption \eqref{escMAINRATNERTHM1res1} in Theorem \ref{escMAINRATNERTHM1}
immediately implies that
\begin{align}\label{MUTOT0lempf1}
\lim_{\rho\to0}\lambda(X_\rho^\vecm)
=0\qquad\text{for each fixed $\vecm\in q^{-1}\Z^d$.}
\end{align}
Next, since $\vecphi$ is Lipschitz and $\scrU$ is bounded
(after the initial reduction on p.\ \pageref{INITIALREDUCTION}), 
there exists a constant $C_1>0$ such that
for any $\rho>0$ and $\vecm\in q^{-1}\Z^d$,
\begin{align}\label{escMAINRATNERTHM1pf1}
|\vece_1\cdot M\vecm |>C_1\:\Rightarrow\:
\Leb\bigl(X_\rho^\vecm\bigr)\ll\rho^{d-1}|\vece_1\cdot M\vecm|^{1-d}.
\end{align}
(Here and in the rest of the proof, the implied constant in any $\ll$ bound is allowed to depend on
$C,q,M,\vecw,\vecphi$, but \textit{not} on $\vecm,\eta,\rho$.)
Furthermore, increasing $C_1$ if necessary,
and assuming that $\rho$ is so small that $\eta\rho^{1-d}\geq1$ and $C\rho<1$, we see that
\begin{align}
|\vece_1\cdot M\vecm|\geq C_1\eta\rho^{1-d}\:\Rightarrow\: Y_\rho^\vecm=\emptyset.
\end{align}
and
\begin{align*}
\|(M\vecm)_\perp\|\geq C_1\bigl(1+|\vecm M\cdot\vece_1|\bigr)\:\Rightarrow\: X_\rho^\vecm=\emptyset.
\end{align*}
Hence if we set 
\begin{align*}
&A_1=\{\vecm\in q^{-1}\Z^d\col|\vece_1\cdot M\vecm|<C_1\eta\rho^{1-d}\};
\\
&A_2=\{\vecm\in q^{-1}\Z^d\col|\vece_1\cdot M\vecm|>C_1\};
\\
&A_3=\{\vecm\in q^{-1}\Z^d\col\|(\vecm M)_\perp\|< C_1(1+|\vece_1\cdot M\vecm|)\},
\end{align*}
then for any $0<\eta<1$ and $0<\rho<\min(C^{-1},\eta^{1/(d-1)})$, we have
\begin{align*}
\lambda\biggl(\bigcup_{\vecm\in q^{-1}\Z^d}\tT_\rho^\vecm\biggr)\leq
\sum_{\vecm\in A_1\cap A_3}\lambda\bigl(X_\rho^\vecm\bigr)
\ll\sum_{\vecm\in A_1\cap A_2\cap A_3}\rho^{d-1}|\vece_1\cdot M\vecm|^{1-d}
+\sum_{\vecm\in A_3\setminus A_2}\lambda(X_\rho^\vecm).
\end{align*}
(In the last bound we used the fact that $\lambda(A)\ll\Leb(A)$ uniformly over all Borel sets $A\subset\scrU$,
because of our initial reduction on p.\ \pageref{INITIALREDUCTION}.)
Here $A_3\setminus A_2$ is a finite set,
and hence the last sum above tends to zero as $\rho\to0$, by \eqref{MUTOT0lempf1}.
Finally the set $A_1\cap A_2\cap A_3$ can be covered by the dyadic pieces
$D_s=A_3\cap\{ 2^s C_1<|\vece_1\cdot M\vecm|\leq 2^{s+1}C_1\}$ with $s$ running through
$0,1,\ldots,S:=\lceil\log_2(\eta\rho^{1-d})\rceil$.
Here $\# D_s\ll2^{sd}$ and so
\begin{align*}
\sum_{\vecm\in A_1\cap A_2\cap A_3}\rho^{d-1}|\vece_1\cdot M\vecm|^{1-d}
\ll\rho^{d-1}\sum_{s=0}^S 2^{sd}\cdot 2^{s(1-d)}
\ll\rho^{d-1}2^{S}\ll\eta.
\end{align*}
Taken together these bounds prove that \eqref{MUTOT0lempf2} holds, and the lemma is proved.
\end{proof}

We are now in a position to complete the proof of Theorem \ref{escMAINRATNERTHM1}.

\begin{proof}[Conclusion of the proof of Theorem \ref{escMAINRATNERTHM1}]
We wish to prove that our arbitrary weak* limit $\mu$ necessarily equals $\mu_X$.
Assume the contrary; $\mu\neq\mu_X$;
then in the ergodic decomposition \eqref{ERGDEC} we have $P(\scrE\setminus\{\mu_{X}\})>0$.
Using then Lemma \ref{SUPPORTlem}, and the fact that there are only countably many $q\in\Z^+$,
and countably many subspaces $U\subset\R^k$ intersecting $\Z^k$ in a lattice,
it follows that there exists some such subspace $U$ of dimension $<k$, and some $q\in\Z^+$, such that
$\mu\bigl(\:\bigcup\:\bigl\{\scrX_{U,q,\vecxi}\col\vecxi\in(\vece_1^\perp)^k\bigr\}\bigr)>0$.
This contradicts Lemma \ref{MUTOT0lem}.
Hence Theorem~\ref{escMAINRATNERTHM1} is proved.
\label{escMAINRATNERTHM1pfFINISH}
\end{proof}

Next we note the following consequence of Theorem \ref{escMAINRATNERTHM1}.
\begin{cor}\label{SPHERICALRATNERcor}
Let $M\in\SL(d,\R)$, let $\scrU\subset\R^{d-1}$ be an open subset and let $E_1:\scrU\to\SO(d)$ be a smooth map such that the map
$\vecx\mapsto E_1(\vecx)^{-1}\vece_1$ from $\scrU$ to $\S_1^{d-1}$ has a nonsingular differential at (Lebesgue-)almost
all $\vecx\in \scrU$.
Let $\vecphi:\scrU\to(\R^d)^k$ be a Lipschitz map,
and let $\lambda$ be a Borel probability measure on $\scrU$, absolutely continuous with respect to Lebesgue measure.
Assume that for every $\vecw=(w_1,\ldots,w_k)\in\Z^k\setminus\{\bn\}$,
\begin{align}\label{SPHERICALRATNERcorASS1}
\lambda\biggl(\biggl\{\vecx\in \scrU\col
\sum_{j=1}^k w_j\cdot\vecphi_j(\vecx)\in\R M^{-1}E_1(\vecx)^{-1}\vece_1+\Q^d\biggr\}\biggr)=0.
\end{align}
Then for any bounded continuous function $f:X\times\scrU\to\R$,
\begin{align}
\lim_{\rho\to0}\int_{\scrU}
f\bigl(D(\rho) E_1(\vecx)M(1_d,\vecphi(\vecx)),\vecx\bigr)\,d\lambda(\vecx)
=\int_{X\times\scrU}f(g,\vecx)\,d\mu_X(g)\,d\lambda(\vecx).
\end{align}
\end{cor}

\begin{proof}
Let us first note that if 
\eqref{escMAINRATNERTHM1res2} holds for every bounded continuous function
$f:X\to\R$, then by a standard approximation argument
(cf.\ \cite[proof of Thm.\ 5.3]{partI}),
also the following more general limit statement holds:
For each small $\rho>0$,
let $f_\rho:X\times\scrU\to\R$ be a continuous function
satisfying $|f_\rho|<B$
where $B$ is a fixed constant,
and assume that $f_\rho\to f$ as $\rho\to0$, uniformly on compacta,
for some continuous function $f:X\times\scrU\to\R$.
Then
\begin{align}\label{escMAINRATNERTHM1res2gen}
\lim_{\rho\to0}\int_{\scrU}
f_\rho\bigl(D(\rho) n_-(\vecv) M (1_d,\vecphi(\vecv)),\vecv\bigr)\,d\lambda(\vecv)
=\int_{X\times\scrU}f(g,\vecv)\,d\mu_X(g)\,d\lambda(\vecv).
\end{align}
Now Corollary \ref{SPHERICALRATNERcor} is proved by a direct mimic of the proof of
\cite[Cor.\ 5.4]{partI},
using \eqref{escMAINRATNERTHM1res2gen} in place of 
\cite[Thm.\ 5.3]{partI}.
(Recall that we translate from the setting in \cite{partI}
by applying the transpose map, which also changes order of multiplication.
Following the proof of 
\cite[\href{http://file://T:marklof/marklofstrombergsson2010b.pdf:27}{Cor.\ 5.4}]{partI},
the task becomes to prove that
$D(\rho) n_-(\widetilde\vecx)E_1(\vecx_0)M(1_d,\vecphi(\vecx))$,
for $\vecx$ in a fixed small neighborhood of an arbitrary point $\vecx_0\in\scrU$,
becomes asymptotically equidistributed in $X$ as $\rho\to0$.
Here 
$\widetilde\vecx=-c(\vecx)^{-1}\vecv(\vecx)$ with
$c(\vecx)$ and $\vecv(\vecx)$ given by
$\begin{pmatrix}c(\vecx)\\\vecv(\vecx)\end{pmatrix}=E_1(\vecx_0)E_1(\vecx)^{-1}\vece_1$.
The condition for equidistribution, \eqref{escMAINRATNERTHM1res1}, then becomes
\begin{align*}
\lambda\biggl(\biggl\{\vecx\col
\sum_{j=1}^k w_j\cdot\vecphi_j(\vecx)\in\R M^{-1}E_1(\vecx_0)^{-1}\begin{pmatrix}1\\-\widetilde\vecx\end{pmatrix}
+\Q^d\biggr\}\biggr)=0,
\end{align*}
or equivalently, \eqref{SPHERICALRATNERcorASS1}.)
\end{proof}

Finally from Corollary \ref{SPHERICALRATNERcor} we derive the following equidistribution result,
which is more directly adapted to the proof of Theorem \ref{thm:main001}.
Recall from Section \ref{sec:general} that we have fixed the map
$\vecv\mapsto R_\vecv$, $\US\to\SO(d)$,
such that $R_\vecv \vecv=\vece_1$ for all $\vecv\in\US$,
and such that $\vecv\mapsto R_\vecv$ is smooth throughout $\US\setminus\{\vecv_0\}$.
Note that since the proof below involves using Sard's Theorem, 
the proof does not apply to arbitrary Lipschitz maps.

\begin{thm}\label{KEYEQUIDISTRTHM2}
Let $\scrU$ be an open subset of $\R^m$ ($m\geq1$),
let $\lambda$ be a Borel probability measure on $\scrU$ which is absolutely continuous with respect to Lebesgue measure,
and let $\vecf:\scrU\to \RR^d$ be a smooth map.
Assume that $\vecf(\vecJ)\neq\bn$ for all $\vecJ\in\scrU$ and $\lambda$ is $\vecf$-regular.
Also let $\vecphi:\scrU\to(\R^d)^k$ be a smooth map such that
for every $\vecm=(m_1,\ldots,m_k)\in\Z^k\setminus\{\bn\}$,
\begin{align}\label{KEYEQUIDISTRTHM2ass}
\lambda\bigg(\bigg\{ \vecJ\in\scrU : \sum_{j=1}^k m_j \, \vecphi_j(\vecJ)
\in \RR \vecf(\vecJ) + \QQ^d \bigg\}\bigg) = 0.
\end{align}
Then for any $h\in\C_b(X\times\scrU)$, writing $\vecv(\vecJ):=\|\vecf(\vecJ)\|^{-1}\vecf(\vecJ)$,
\begin{align}\label{KEYEQUIDISTRTHM2res}
\lim_{\rho\to0}
\int_{\scrU}h\big(D(\rho) R_{\vecv(\vecJ)}\big(1_d,\vecphi(\vecJ) 
\big),\,\vecJ\bigr)\,d\lambda(\vecJ)
=\int_{\scrU}\int_{X}h(p,\vecJ)\,d\mu_X(p)\,d\lambda(\vecJ).
\end{align}
\end{thm}
\begin{proof} 
Note that $\vecv$ is a smooth map from $\scrU$ to $\US$,
and the fact that $\lambda$ is $\vecf$-regular means exactly that
$\vecv_*(\lambda)$ is absolutely continuous with respect to the Lebesgue measure on $\US$.
Hence $m\geq d-1$, 
and by Sard's Theorem the set of critical values of $\vecv$ has measure zero with respect to $\vecv_*(\lambda)$,
and so the set of critical points of $\vecv$ has measure zero with respect to $\lambda$.
For each point $\vecJ\in\scrU$ which is not a critical point of $\vecv$,
there exists a diffeomorphism $\iota$ from the unit box $(0,1)^m$ onto an open neighborhood of $\vecJ$ in $\scrU$
such that $\vecv(\iota(\vecx))$ depends only on $(x_1,\ldots,x_{d-1})$,
and this function gives a diffeomorphism of $(0,1)^{d-1}$ onto an open subset of $\US$.
Hence by decomposition and approximation of $\lambda$,
it follows that it suffices to prove Theorem \ref{KEYEQUIDISTRTHM2} 
in the case when $\lambda$ is supported in a \textit{fixed} such coordinate neighborhood.
Changing coordinates via the diffeomorphism $\iota$,
we may assume from now on that $\scrU=(0,1)^m$ and that
$\vecv(\vecx)$ depends only on $(x_1,\ldots,x_{d-1})$
and gives a diffeomorphism of $(0,1)^{d-1}$ onto an open subset of $\US$.

Let us first assume $m=d-1$.
Then $\vecv$ is a diffeomorphism of $\scrU=(0,1)^{d-1}$ onto an open subset of $\US$.
Recall that $\vecv\mapsto R_\vecv$ is smooth throughout $\US\setminus\{\vecv_0\}$.
If $\vecv_0$ is in the image of $\vecv$, then we replace $\scrU$ by $\scrU\setminus\vecv^{-1}(\vecv_0)$.
Now the map $\vecx\mapsto R_{\vecv(\vecx)}$ is smooth throughout $\scrU$,
and $\vecx\mapsto R_{\vecv(\vecx)}^{-1}\vece_1=\vecv(\vecx)$ has everywhere nonsingular differential.
Now \eqref{KEYEQUIDISTRTHM2res} follows from Corollary \ref{SPHERICALRATNERcor} applied with $M=1_d$
and $E_1(\vecx)=R_{\vecv(\vecx)}$.

It remains to consider the case $m>d-1$.
We are assuming that $\lambda$ is absolutely continuous;
hence $\lambda$ has a density $\lambda'\in\L^1((0,1)^m,d\vecx)$.
Now \eqref{KEYEQUIDISTRTHM2ass} says that
\begin{align}\notag
\int_{(0,1)^m}I\biggl(\sum_{j=1}^k m_j\vecphi_j(\vecx)\in\R\vecv(\vecx)+\Q^d\biggr)
\,\lambda'(\vecx)\,d\vecx=0.
\end{align}
Decompose $\vecx$ as $(\vecx_1,\vecx_2)\in\R^{d-1}\times\R^{m-d-1}$,
and recall that $\vecv(\vecx)$ only depends on $\vecx_1$, i.e.\ we may write $\vecv(\vecx)=\vecv(\vecx_1)$.
It follows that for (Lebesgue) a.e.\ $\vecx_2\in(0,1)^{m-d-1}$,
\begin{align*}
\int_{(0,1)^{d-1}}I\biggl(\sum_{j=1}^k m_j\vecphi_j(\vecx_1,\vecx_2)
\in\R\vecv(\vecx_1)+\Q^d\biggr)
\,\lambda'(\vecx_1,\vecx_2)\,d\vecx_1=0.
\end{align*}
Furthermore $\int_{(0,1)^m}\lambda'(\vecx_1,\vecx_2)\,d\vecx_1\,d\vecx_2=1$;
hence for a.e.\ $\vecx_2$ we have
$\int_{(0,1)^{d-1}}\lambda'(\vecx_1,\vecx_2)\,d\vecx_1<\infty$.
For each fixed $\vecx_2\in(0,1)^{m-d-1}$ which satisfies both the last two conditions,
our result for the case $m=d-1$ applies, showing that 
\begin{align*}
\lim_{\rho\to0}\int_{(0,1)^{d-1}}h_1\big(D(\rho) R_{\vecv(\vecx_1)}(1_d,\vecphi(\vecx_1,\vecx_2)),(\vecx_1,\vecx_2)\big)
\,\lambda'(\vecx_1,\vecx_2)\,d\vecx_1
\\
=\int_{(0,1)^{d-1}\times X} h_1\big(p,(\vecx_1,\vecx_2))\,\lambda'(\vecx_1,\vecx_2)\,d\vecx_1\,d\mu_X(p).
\end{align*}
Now \eqref{KEYEQUIDISTRTHM2res} follows by integrating the last relation
over $\vecx_2\in(0,1)^{m-d-1}$, applying Lebesgue's Bounded Convergence Theorem
to change order of limit and integration.
\end{proof}

\section{Proof of Theorem \ref*{thm:main001}}
\label{MAINPROOFsec}


We now give the proof of Theorem \ref{thm:main001}.
We will only discuss the proof of \eqref{thm:main001res2} in detail.
The proof of \eqref{thm:main001res1} is completely similar;
basically one just has to replace $\osigma^{(k)}(\vecJ)$ with the constant
$\osigma_\lambda^{(k)}$ throughout the discussion; cf.\ Remark \ref{thm:main001res1pfrem} below.

Recall that 
\begin{align}
\vecv(\vecJ)=\frac{\vecf(\vecJ)}{\|\vecf(\vecJ)\|}\in\US
\qquad (\vecJ\in\scrU).
\end{align}

We start by making some initial reductions.
First, the assumptions of Theorem \ref{thm:main001} 
imply that the open subset
\begin{align}\label{thm:main001pf1}
\{\vecJ\in\scrU\col\vecv(\vecJ)\neq\vecv_0,\:\vecu_j(\vecJ)\neq\vecv_0\:\forall j
\}
\end{align}
has full measure in $\scrU$ with respect to $\lambda$,
and so we may just as well replace $\scrU$ by that set. 
Hence from now on $R_{\vecv(\vecJ)}$ is a smooth function on all $\scrU$,
and the same holds for $R_{\vecu_j(\vecJ)}$ for each $j\in\{1,\ldots,k\}$.
Next let us set, for $\eta>0$,
\begin{align}\label{Ueta1}
\scrU_\eta:=\{\vecJ\in\scrU\col\|\vecphi_j(\vecJ)-\vecphi_\ell(\vecJ)\|>\eta\:\forall j\neq\ell\},
\end{align}
where $\|\cdot\|$ denotes distance to the origin in $\TT^d$
(viz., $\|\vecx\|=\inf_{\vecm\in\Z^d}\|\tvecx-\vecm\|$ for any $\vecx\in\TT^d$,
where $\tvecx$ is any lift of $\vecx$ to $\R^d$).
Note that the fact that $(\vecphi_1,\ldots,\vecphi_k)$ is $(\vectheta,\lambda)$-generic
implies that for any $j\neq\ell$, 
$\vecphi_j(\vecJ)\neq\vecphi_\ell(\vecJ)$ holds for $\lambda$-a.e.\ $\vecJ\in\scrU$.
Hence $\lambda(\scrU_\eta)\to1$ as $\eta\to0$,
and thus by a standard approximation argument
(cf., e.g., \cite[Thm.\ 4.28]{kallenberg02}),
it suffices to prove that for all sufficiently small $\eta>0$, the convergence \eqref{thm:main001res2}
holds when $\scrU$ is replaced by $\scrU_\eta$ and $\lambda$ is replaced by
$\lambda(\scrU_\eta)^{-1}\lambda_{|\scrU_\eta}$.
In other words, from now on we may assume that there exists a constant $0<\eta<1$ such that
$\|\vecphi_j(\vecJ)-\vecphi_\ell(\vecJ)\|>\eta$ for all $j\neq\ell$ and $\vecJ\in\scrU$.

For any $j\in\{1,\ldots,k\}$, $\rho>0$, $T>0$, we
introduce the following ``cylinder'' subset of $\R^d\times\scrU$:
\begin{align}\label{AjrhoTdef}
&A_{j,\rho,T}:=\biggl\{\biggl(t\vecf(\vecJ)-\rho R_{\vecu_j(\vecJ)}^{-1}\cmatr{0}{\vecx},\vecJ\biggr)\:\bigg|\:
(\vecx,\vecJ)\in\Omega_j,\:0<t\leq T\osigma^{(k)}(\vecJ)\rho^{1-d}\biggr\}.
\end{align}
For any subset $A\subset\R^d\times\scrU$ and $\vecJ\in\scrU$,
we write $A(\vecJ):=\{\vecx\in\R^d\col(\vecx,\vecJ)\in A\}$.
Let us set
\begin{align}\label{Cdef}
C:=\sup\bigl\{\|\vecx\|\col j\in\{1,\ldots,k\},\:(\vecx,\vecJ)\in\Omega_j\bigr\};
\end{align}
this is a finite positive real constant, since each $\Omega_j$ is a non-empty bounded open set.
\begin{lem}\label{MAINPFlem1}
For any $0<\rho<\eta/(10C)$, $(\vectheta,\vecJ)\in\TT^d\times\scrU$, $n\in\Z^+$ and $T>0$, 
the following equivalence holds:
\begin{align}\label{MAINPFlem1res1}
&
\frac{\rho^{d-1}t_n(\vectheta,\vecJ,\scrD_\rho^{(k)})}{\osigma^{(k)}(\vecJ)}\leq T
\quad\Leftrightarrow\quad
\sum_{j=1}^k\#\bigl(A_{j,\rho,T}(\vecJ)\cap(\vecphi_j(\vecJ)-\vectheta+\Z^d)\bigr)\geq n.
\end{align}
\end{lem}
(In \eqref{MAINPFlem1res1}, $\vecphi_j(\vecJ)-\vectheta+\Z^d$ denotes a translate of the lattice $\Z^d$,
i.e.\ a subset of $\R^d$. 
Note that this set is well-defined, i.e.\ independent of the choice of lifts of
$\vecphi_j(\vecJ)$ and $\vectheta$ to $\R^d$.)
\begin{proof}
Let $\rho$, $(\vectheta,\vecJ)$, $n$ and $T$ be given as in the statement of the lemma.
Note that the given restriction on $\rho$ implies that each target set,
\begin{align}\label{MAINPFlem1pf2}
\scrD_\rho(\vecu_j,\vecphi_j,\Omega_j)(\vecJ)=\biggl\{
\vecphi_j(\vecJ)+\rho R_{\vecu_j(\vecJ)}^{-1}\cmatr{0}{\vecx}\:\bigg|\:\vecx\in\Omega_j(\vecJ)\biggr\}
\subset\TT^d
\end{align}
is contained within a ball of radius $<\eta/10<1/10$, 
centered at $\vecphi_j(\vecJ)$.
In particular each target is injectively embedded in $\TT^d$,
and the targets for $j=1,\ldots,k$ are pairwise disjoint, since 
$\|\vecphi_j(\vecJ)-\vecphi_\ell(\vecJ)\|>\eta$ for all $j\neq\ell$. 
Hence the left inequality in \eqref{MAINPFlem1res1} holds if and only if
\begin{align}\label{MAINPFlem1pf1}
\sum_{j=1}^k\#\biggl\{t\in\bigl(0,T\osigma^{(k)}(\vecJ)\rho^{1-d}\bigr]\col
\vectheta+t\vecf(\vecJ)\in\scrD_\rho(\vecu_j,\vecphi_j,\Omega_j)(\vecJ)\biggr\}\geq n.
\end{align}
Note that each set in the left hand side is a discrete set of $t$-values, since the target set
$\scrD_\rho(\vecu_j,\vecphi_j,\Omega_j)(\vecJ)$ is contained in a hyperplane orthogonal to $\vecu_j(\vecJ)$,
and $\vecu_j(\vecJ)\cdot\vecf(\vecJ)>0$ by assumption.
Lifting the situation from $\TT^d$ to $\R^d$ we now see,
via \eqref{MAINPFlem1pf2} and \eqref{AjrhoTdef}, that
for each $j$ the corresponding term in the left hand side of \eqref{MAINPFlem1pf1} equals
$\#\bigl(A_{j,\rho,T}(\vecJ)\cap(\vecphi_j(\vecJ)-\vectheta+\Z^d)\bigr)$.
Hence the lemma follows.
\end{proof}

Next we prove that the linear map
$D(\rho)R_{\vecv(\vecJ)}$ takes the cylinder $A_{j,\rho,T}(\vecJ)$ 
into a cylinder which is approximately \textit{normalized,} 
in an appropriate sense.
Indeed, for any real numbers $Y<Z$, define $\tA_{j,Y,Z}\subset\R^d\times\scrU$ 
through
\begin{align}\label{tApdef}
\tA_{j,Y,Z}:=\biggl\{\biggl(\cmatr{t}{-\tR_j(\vecJ)\vecx},\vecJ\biggr)\:\bigg|\:
(\vecx,\vecJ)\in\Omega_j,\:\osigma^{(k)}(\vecJ)\|\vecf(\vecJ)\|Y<t\leq \osigma^{(k)}(\vecJ)\|\vecf(\vecJ)\|Z\biggr\},
\end{align}
where $\tR_j(\vecJ)$ is as on p.\ \pageref{tRjJdef}.
We then have the following lemma.
\begin{lem}\label{AjrhoTapprLEM}
Given $\ve>0$ and $T>0$,
there exists $\rho_0>0$ such that for all $\rho\in(0,\rho_0)$, $j\in\{1,\ldots,k\}$ and $\vecJ\in\scrU$,
\begin{align*}
\tA_{j,\ve,T-\ve}(\vecJ)\subset D(\rho)R_{\vecv(\vecJ)}A_{j,\rho,T}(\vecJ)\subset \tA_{j,-\ve,T+\ve}(\vecJ)
\end{align*}
\end{lem}
\begin{proof}
By direct computation,
\begin{align*}
D(\rho)R_{\vecv(\vecJ)}A_{j,\rho,T}(\vecJ)
=\biggl\{t\vece_1-\rho D(\rho)\fR_j(\vecJ)\cmatr{0}{\vecx}\:\bigg|\:
\vecx\in\Omega_j(\vecJ),\:0<t\leq \osigma^{(k)}(\vecJ)\|\vecf(\vecJ)\| T\biggr\}.
\end{align*}
Using $\rho D(\rho)=\diag(\rho^d,1,\ldots,1)$ and 
\eqref{Cdef}, it follows that for every $\vecx\in\Omega_j(\vecJ)$,
\begin{align*}
\rho D(\rho)\fR_j(\vecJ)\cmatr{0}{\vecx}
=\cmatr{r}{\tR_j(\vecJ)\vecx}
\end{align*}
where $|r|\leq C\rho^d$.
Note also that, by \eqref{osigmakformula},
\begin{align*}
\bigl(\osigma^{(k)}(\vecJ)\|\vecf(\vecJ)\|\bigr)^{-1}=
\sum_{j=1}^k\Leb(\Omega_j(\vecJ))\,\vecu_j(\vecJ)\cdot\vecv(\vecJ)
\end{align*}
and this sum is bounded from above by a constant independent of $\vecJ$, since each set $\Omega_j$ is bounded.
The lemma follows from these observations.
\end{proof}

Let $G_1=\SL(d,\R)\ltimes\R^d$.
This is the group ``$G$ for $k=1$'';
in particular $G_1$ acts on $\R^d$ (cf.\ \eqref{GACTRd}).
For $g=(M,(\vecxi_1,\ldots,\vecxi_k))\in G$ and $j\in\{1,\ldots,k\}$
we write $g^{[j]}:=(M,\vecxi_j)\in G_1$.
We also introduce the short-hand notation $\oN:=\{1,\ldots,N\}$.
Given real numbers $Y_n<Z_n$ for $n\in\oN$, we define $B[(Y_n),(Z_n)]$ to be the following subset of $X\times\scrU$:
\begin{align}\label{Bdef}
&B[(Y_n),(Z_n)]:=\biggl\{(g\Gamma,\vecJ)\in X\times\scrU\col
\sum_{j=1}^k\#\bigl(\tA_{j,Y_n,Z_n}(\vecJ)\cap g^{[j]}(\Z^d)\bigr)\geq n
\hspace{8pt}\forall n\in\oN\biggr\}.
\end{align}
In the following the Lebesgue measure in various dimensions will appear within the same discussion;
for clarity we will therefore write $\Leb_m$ for the Lebesgue measure in $\R^m$.

The following is a ``trivial'' variant of Siegel's mean value theorem \cite{Siegel}:
\begin{lem}\label{SIEGELLEM}
For any $j\in\{1,\ldots,k\}$ and $f\in\L^1(\R^d)$,
\begin{align}\label{SIEGELLEMres}
\int_X\sum_{\vecm\in\Z^d}f(g^{[j]}(\vecm))\,d\mu_X(g)
=\int_{\R^d}f(\vecx)\,d\vecx.
\end{align}
In particular for any Lebesgue measurable subset $A\subset\R^d$,
\begin{align}\label{SIEGELLEMres2}
\mu_X(\{\Gamma g\in X\col g^{[j]}(\Z^d)\cap A\neq\emptyset\})\leq\Leb_d(A).
\end{align}
\end{lem}
\begin{proof}
(Cf., e.g., \cite[proof of Lemma 10]{SV}.)
In the left hand side of \eqref{SIEGELLEMres} we write $g=(M,(\vecxi_1,\ldots,\vecxi_k))$,
integrate out all variables $\vecxi_\ell$, $\ell\neq j$,
and then substitute $\vecxi_j=M\veceta$; this gives
\begin{align}\label{SIEGELLEMpf1}
\int_X\sum_{\vecm\in\Z^d}f(g^{[j]}(\vecm))\,d\mu_X(g)=
\int_F\int_{[0,1]^d}\sum_{\vecm\in\Z^d}f(M(\vecm+\veceta))\,d\veceta\,d\mu(M),
\end{align}
where $F\subset\SL_d(\R)$ is a fundamental domain for $\SL_d(\R)/\SL_d(\Z)$
and $\mu$ is Haar measure on $\SL_d(\R)$ normalized so that $\mu(F)=1$.
Now \eqref{SIEGELLEMres} follows since the inner integral in \eqref{SIEGELLEMpf1} 
equals $\int_{\R^d}f(\vecx)\,d\vecx$ for every $M$.
The last statement of the lemma then follows by noticing that the left hand side of 
\eqref{SIEGELLEMres2} is bounded above by
the left hand side of \eqref{SIEGELLEMres} with $f$ equal to the characteristic function of $A$.
\end{proof}

\begin{lem}\label{BBpvolsamecontLEM}
The number $(\mu_X\times\lambda)\bigl(B[(Y_n),(Z_n)]\bigr)$ depends continuously on $((Y_n),(Z_n))$.
\end{lem}

\vspace{-5pt}

(Here we keep $((Y_n),(Z_n))\in\R^N\times\R^N$ subject to $Y_n<Z_n$ for all $n\in\oN$, as before.)

\begin{proof}
Let $\fD(\vecJ)\in\SL(d,\R)$ be the diagonal matrix
\begin{align*}
\fD(\vecJ)=\diag\Bigl[\bigl(\osigma^{(k)}(\vecJ)\|\vecf(\vecJ)\|\bigr)^{-1},
\bigl(\osigma^{(k)}(\vecJ)\|\vecf(\vecJ)\|\bigr)^{1/(d-1)},\ldots,\bigl(\osigma^{(k)}(\vecJ)\|\vecf(\vecJ)\|\bigr)^{1/(d-1)}
\Bigr].
\end{align*}
Using the fact that $\mu_X$ is $G$-invariant
(thus invariant under $g\Gamma\mapsto\fD(\vecJ)g\Gamma$)
we see that 
\begin{align}\label{BBpvolsamecontLEMpf1}
(\mu_X\times\lambda)\bigl(B[(Y_n),(Z_n)]\bigr)=(\mu_X\times\lambda)\bigl(B'[(Y_n),(Z_n)]\bigr),
\end{align}
where $B'[(Y_n),(Z_n)]$ is the set obtained by replacing
$\tA_{j,Y,Z}(\vecJ)$ by $\tA_{j,Y,Z}'(\vecJ):=\fD(\vecJ)\tA_{j,Y,Z}(\vecJ)$ in the definition \eqref{Bdef}.
Hence it now suffices to prove that $(\mu_X\times\lambda)\bigl(B'[(Y_n),(Z_n)]\bigr)$ depends continuously on $((Y_n),(Z_n))$.
Note also that
\begin{align}\label{tAjYZ}
\tA'_{j,Y,Z}(\vecJ):=\biggl\{\cmatr{t}{-\vecx}\:\bigg|\:\vecx\in\tOmega_j(\vecJ),\:Y<t\leq Z\biggr\},
\end{align}
where $\tOmega_j(\vecJ)$ is as in \eqref{tOmegadef2}.

To prove the continuity, consider any real numbers $Y_n,Z_n,Y_n',Z_n'$ for $n\in\oN$,
subject to $Y_n<Z_n$ and $Y_n'<Z_n'$.
Writing $\triangle$ for symmetric set difference, we have
\begin{align*}
&B'[(Y_n),(Z_n)]\:\triangle\: B'[(Y'_n),(Z'_n)]
\\
&\subset\bigcup_{n\in\oN}\bigcup_{j=1}^k\biggl\{(g\Gamma,\vecJ)\in X\times\scrU\col
\bigl(\tA'_{j,Y_n,Z_n}(\vecJ)\:\triangle\:\tA'_{j,Y'_n,Z'_n}(\vecJ)\bigr)\cap g^{[j]}(\Z^d)\neq\emptyset\biggr\},
\end{align*}
and hence by \eqref{SIEGELLEMres2} 
and \eqref{tAjYZ},
\begin{align*}
(\mu_X\times\lambda)\bigl(B'[(Y_n),(Z_n)]\:\triangle\: B'[(Y'_n),(Z'_n)]\bigr)
\leq\sum_{n\in\oN}\sum_{j=1}^k\int_{\scrU}\Leb_d\bigl(\tA_{j,Y_n,Z_n}'(\vecJ)\:\triangle\:\tA'_{j,Y'_n,Z'_n}(\vecJ)\bigr)
\,d\lambda(\vecJ)
\\
\leq\sum_{n\in\oN}\sum_{j=1}^k\int_{\scrU}\Leb_1\Bigl((Y_n,Z_n]\:\triangle\:(Y_n',Z_n']\Bigr)
\Leb_{d-1}(\tOmega_j(\vecJ))\,d\lambda(\vecJ).
\end{align*}
However it follows from \eqref{dettRjJ} and \eqref{tOmegadef2} that
\begin{align*}
\Leb_{d-1}(\tOmega_j(\vecJ))=\osigma^{(k)}(\vecJ)\Leb_{d-1}(\Omega_j(\vecJ))\,\vecu_j(\vecJ)\cdot\vecf(\vecJ),
\end{align*}
and using also \eqref{osigmakformula} it follows that
\begin{align}\label{TOTVOL1}
\sum_{j=1}^k\Leb_{d-1}(\tOmega_j(\vecJ))=1
\end{align}
for all $\vecJ\in\scrU$.
Hence we conclude
\begin{align*}
&\Bigl|(\mu_X\times\lambda)\bigl(B'[(Y_n),(Z_n)]\bigr)-(\mu_X\times\lambda)\bigl(B'[(Y'_n),(Z'_n)]\bigr)\Bigr|
\leq\sum_{n\in\oN}\bigl|Y_n-Y_n'\bigr|+\sum_{n\in\oN}\bigl|Z_n-Z_n'\bigr|.
\end{align*}
This proves the desired continuity.
\end{proof}

We wish to prove that the limit relation \eqref{KEYEQUIDISTRTHM2res} in Theorem \ref{KEYEQUIDISTRTHM2}
holds with $h$ equal to the characteristic function of $B=B[(Y_n),(Z_n)]$.
For this we need to prove that the boundary, $\partial B$, has measure zero with respect to $\mu_X\times\lambda$.
Here by $\partial B$ we denote the boundary of $B$ \textit{in} $X\times\scrU$,
and similarly $\partial\tA_{j,Y_n,Z_n}$ denotes the boundary of $\tA_{j,Y_n,Z_n}$ \textit{in} $\R^d\times\scrU$.
(The alternative would have been to consider the boundaries in $X\times\R^m$ and $\R^d\times\R^m$, respectively.)
\begin{lem}\label{BDRYlem}
For any $B=B[(Y_n),(Z_n)]$, if $(g\Gamma,\vecJ)\in\partial B$ then 
$g^{[j]}(\Z^d)\cap(\partial\tA_{j,Y_n,Z_n})(\vecJ)\neq\emptyset$
for some $j\in\{1,\ldots,k\}$ and $n\in\oN$.
\end{lem}
\begin{proof}
Assume $(g\Gamma,\vecJ)\in\partial B$.
Then there exist sequences $\{(g_m\Gamma,\vecJ_m)\}$ 
and $\{(\tg_m\Gamma,\tJ_m)\}$ in $X\times\scrU$ 
such that both $(g_m\Gamma,\vecJ_m)\to(g\Gamma,\vecJ)$ and $(\tg_m\Gamma,\tJ_m)\to(g\Gamma,\vecJ)$ as $m\to\infty$,
and $(g_m\Gamma,\vecJ_m)\in B$ and $(\tg_m\Gamma,\tJ_m)\notin B$ for all $m$.
In particular for each $m$ there is some $n\in\oN$ such that
\begin{align}\label{BDRYlempf1}
\sum_{j=1}^k\#\bigl(\tA_{j,Y_n,Z_n}(\tJ_m)\cap \tg_m^{[j]}(\Z^d)\bigr)<n.
\end{align}
By passing to an appropriate subsequence, we may in fact assume that $n$ is \textit{fixed} in \eqref{BDRYlempf1},
i.e.\ $n$ does not depend on $m$.
On the other hand $(g_m\Gamma,\vecJ_m)\in B$ for each $m$, and thus
\begin{align}\label{BDRYlempf2}
\sum_{j=1}^k\#\bigl(\tA_{j,Y_n,Z_n}(\vecJ_m)\cap g_m^{[j]}(\Z^d)\bigr)\geq n.
\end{align}
Hence for each $m$ there is some $j\in\{1,\ldots,k\}$ such that
\begin{align}\label{BDRYlempf3}
\#\bigl(\tA_{j,Y_n,Z_n}(\tJ_m)\cap \tg_m^{[j]}(\Z^d)\bigr)
<\#\bigl(\tA_{j,Y_n,Z_n}(\vecJ_m)\cap g_m^{[j]}(\Z^d)\bigr).
\end{align}
By again passing to a subsequence we may assume that also $j$ is independent of $m$.
We have $g_m\Gamma\to g\Gamma$ as $m\to\infty$, and by choosing the $g_m$'s 
appropriately we may even assume $g_m\to g$; similarly we may assume $\tg_m\to g$.
Using now $g_m\to g$ and $\vecJ_m\to\vecJ$ 
together with the fact that $\Omega_j$ is bounded,
it follows that there exists a compact set $C\subset\R^d$ such that
$(g_m^{[j]})^{-1}\tA_{j,Y_n,Z_n}(\vecJ_m)\subset C$ for all $m$,
and in particular the cardinality of
$\tA_{j,Y_n,Z_n}(\vecJ_m)\cap g_m^{[j]}(\Z^d)$ remains the same if we replace
$\Z^d$ with the finite set $C\cap\Z^d$.
Now \eqref{BDRYlempf3} implies that for each $m$ there is some $\vecq\in C\cap\Z^d$ such that
$\tg_m^{[j]}(\vecq)\notin\tA_{j,Y_n,Z_n}(\tJ_m)$
but $g_m^{[j]}(\vecq)\in\tA_{j,Y_n,Z_n}(\vecJ_m)$;
and since $C\cap\Z^d$ is finite we may assume, after passing to a subsequence, that
$\vecq$ is independent of $m$.
Taking now $m\to\infty$ it follows that $(g^{[j]}(\vecq),\vecJ)\in\partial\tA_{j,Y_n,Z_n}$, and the lemma is proved.
\end{proof}

\begin{lem}\label{BBDRYMEAS0lem}
Every set $B=B[(Y_n),(Z_n)]$ satisfies $(\mu_X\times\lambda)(\partial B)=0$.
\end{lem}
\begin{proof}
In view of Lemma \ref{BDRYlem} and \eqref{SIEGELLEMres2} in Lemma \ref{SIEGELLEM},
it suffices to prove that for every $j\in\{1,\ldots,k\}$ and $n\in\oN$,
$\partial\tA_{j,Y_n,Z_n}$ has measure zero with respect to $\Leb_d\times\lambda$.
Recalling \eqref{tApdef} we see that, for any $Y<Z$,
\begin{align*}
\partial\tA_{j,Y,Z}=
\biggl\{\biggl(\cmatr{t}{-\tR_j(\vecJ)\vecx},\vecJ\biggr)\:\bigg|\:
(\vecx,\vecJ)\in\partial\Omega_j,\:
\osigma^{(k)}(\vecJ)\|\vecf(\vecJ)\|\,Y\leq t\leq\osigma^{(k)}(\vecJ)\|\vecf(\vecJ)\|\,Z\biggr\}
\hspace{3pt}
\\
\bigcup\:
\biggl\{\biggl(\cmatr{t}{-\tR_j(\vecJ)\vecx},\vecJ\biggr)\:\bigg|\:
(\vecx,\vecJ)\in\overline{\Omega_j},\:
t\in\bigl\{\osigma^{(k)}(\vecJ)\|\vecf(\vecJ)\|\,Y,\osigma^{(k)}(\vecJ)\|\vecf(\vecJ)\|\,Z\bigr\}\biggr\}.
\end{align*}
Now the claim follows by Fubini's Theorem,
using the assumption from Theorem \ref{thm:main001}
that $\partial\Omega_j$ has measure zero with respect to $\Leb_{d-1}\times\lambda$.
\end{proof}

We are now ready to complete the proof of Theorem \ref{thm:main001}.

\begin{proof}[Conclusion of the proof of Theorem \ref{thm:main001}]
Let $\tphi:\scrU\to(\R^d)^k$ be the map 
$$\vecJ\mapsto\bigl(\vecphi_1(\vecJ)-\vectheta(\vecJ),\ldots,\vecphi_k(\vecJ)-\vectheta(\vecJ)\bigr).$$
Then Theorem~\ref{KEYEQUIDISTRTHM2} applies for our $\scrU,\lambda,\vecf$ and $\tphi$;
in particular, the condition \eqref{KEYEQUIDISTRTHM2ass} holds for $\tphi$ since we assume that 
$(\vecphi_1,\ldots,\vecphi_k)$ is $(\vectheta,\lambda)$-generic.
Now for any fixed set $B=B[(Y_n),(Z_n)]$,
since $(\mu_X\times\lambda)(\partial B)=0$ by Lemma \ref{BBDRYMEAS0lem},
a standard approximation argument (cf., e.g., \cite[Thm.\ 4.25]{kallenberg02})
shows that the conclusion of Theorem~\ref{KEYEQUIDISTRTHM2},
\eqref{KEYEQUIDISTRTHM2res},
applies also for $h=\one_{B}$, the characteristic function of $B$.
In other words, 
\begin{align}\label{PORTMANTEAUappllemres}
\lim_{\rho\to0}\lambda\bigl(\bigl\{\vecJ\in\scrU\col
\bigl(D(\rho)R_{\vecv(\vecJ)}\bigl(1_d,\tphi(\vecJ)
\bigr),\:\vecJ\bigr)\in B\bigr\}\bigr)
=(\mu_X\times\lambda)(B).
\end{align}
Combining this with the definition of $B=B[(Y_n),(Z_n)]$, \eqref{Bdef}, we conclude:
\begin{align}\notag
\lim_{\rho\to0}\lambda\biggl(\biggl\{\vecJ\col
\sum_{j=1}^k\#\bigl(\tA_{j,Y_n,Z_n}(\vecJ)\cap
D(\rho)R_{\vecv(\vecJ)}\bigl(\vecphi_j(\vecJ)-\vectheta(\vecJ)+\Z^d\bigr)\bigr)\geq n\hspace{8pt}
\forall n\in\oN\biggr\}\biggr)
\\[3pt]\label{thm:main001pf2}
=(\mu_X\times\lambda)(B).
\end{align}

Now let positive real numbers $T_1,\ldots,T_n$ be given,
and consider a number $\ve$ subject to $0<\ve<\frac12\min(T_1,\ldots,T_n)$.
Applying \eqref{thm:main001pf2} with $Y_n=\ve$ and $Z_n=T_n-\ve$ we get,
via Lemma~\ref{AjrhoTapprLEM}:
\begin{align}\notag
\liminf_{\rho\to0}\lambda\biggl(\biggl\{\vecJ\col
\sum_{j=1}^k\#\bigl(A_{j,\rho,T_n}(\vecJ)\cap
\bigl(\vecphi_j(\vecJ)-\vectheta(\vecJ)+\Z^d\bigr)\bigr)\geq n\hspace{8pt}
\forall n\in\oN\biggr\}\biggr)
\\[3pt]\label{thm:main001pf3}
\geq(\mu_X\times\lambda)\bigl(B[(\ve)_{n=1}^N,(T_n-\ve)_{n=1}^N]\bigr).
\end{align}
Similarly if we take $Y_n=-\ve$ and $Z_n=T_n+\ve$ then we get
\begin{align}\notag
\limsup_{\rho\to0}\lambda\biggl(\biggl\{\vecJ\col
\sum_{j=1}^k\#\bigl(A_{j,\rho,T_n}(\vecJ)\cap
\bigl(\vecphi_j(\vecJ)-\vectheta(\vecJ)+\Z^d\bigr)\bigr)\geq n\hspace{8pt}
\forall n\in\oN\biggr\}\biggr)
\\[3pt]\label{thm:main001pf4}
\leq(\mu_X\times\lambda)\bigl(B[(-\ve)_{n=1}^N,(T_n+\ve)_{n=1}^N]\bigr).
\end{align}
These relations 
hold for all sufficiently small $\ve>0$;
letting $\ve\to0$ we get, via 
Lemma \ref{BBpvolsamecontLEM},
when also rewriting the left hand side using Lemma \ref{MAINPFlem1}:
\begin{align}\label{thm:main001pf5}
\lim_{\rho\to0}\lambda\biggl(\biggl\{\vecJ\col
\frac{\rho^{d-1}t_n(\vectheta,\vecJ,\scrD_\rho^{(k)})}{\osigma^{(k)}(\vecJ)}\leq T_n\hspace{8pt}\forall n\in\oN\biggr\}\biggr)
=(\mu_X\times\lambda)\bigl(B[(0)_{n=1}^N,(T_n)_{n=1}^N]\bigr).
\end{align}
The fact that \eqref{thm:main001pf5} holds for any $T_1,\ldots,T_N>0$ implies that 
\eqref{thm:main001res2} in Theorem \ref{thm:main001} holds.
\end{proof}

\begin{remark}\label{thm:main001res1pfrem}
As mentioned, the proof of \eqref{thm:main001res1} in Theorem \ref{thm:main001} is completely similar;
in principle one only has to replace $\osigma^{(k)}(\vecJ)$ with the constant
$\osigma_\lambda^{(k)}$ throughout the discussion.
However a couple of extra technicalities appear.
First of all, it may happen that $\osigma_\lambda^{(k)}=\infty$;
however in this case \eqref{thm:main001res1} is trivial, with $\tau_i=0$ for all $i$.
Hence from now on we assume $0<\osigma_\lambda^{(k)}<\infty$.
Secondly, the last steps of the proofs of Lemmata \ref{AjrhoTapprLEM} and \ref{BBpvolsamecontLEM} do not carry over
verbatim.
One way to manage those steps is to assume from start that 
$0<\eta<\|\vecf(\vecJ)\|<\eta^{-1}$ for all $\vecJ\in\scrU$;
this is permissible by the argument given below \eqref{Ueta1}, but with $\scrU_\eta$ replaced with
\begin{align}\label{Ueta2}
\scrU_\eta:=\{\vecJ\in\scrU\col\|\vecphi_j(\vecJ)-\vecphi_\ell(\vecJ)\|>\eta\:\forall j\neq\ell\:\text{ and }\:
\eta<\|\vecf(\vecJ)\|<\eta^{-1}\}.
\end{align}
With this assumption,
we have $\bigl(\osigma_\lambda^{(k)}\|\vecf(\vecJ)\|\bigr)^{-1}<\bigl(\osigma_\lambda^{(k)}\eta\bigr)^{-1}$
for all $\vecJ\in\scrU$, and using this the proof of Lemma \ref{AjrhoTapprLEM} extends to the present situation.
Furthermore, by \eqref{oOmegadef1} and \eqref{dettRjJ},
\begin{align*}
\Leb_{d-1}\bigl(\oOmega_j(\vecJ)\bigr)=
\bigl(\osigma^{(k)}_\lambda\,\vecu_j(\vecJ)\cdot\vecf(\vecJ)\bigr)\Leb_{d-1}\bigl(\Omega_j(\vecJ)\bigr)
<\osigma^{(k)}_\lambda \eta^{-1}\Leb_{d-1}\bigl(\Omega_j(\vecJ)\bigr),
\end{align*}
which is bounded from above by a constant independent of $\vecJ$, since the set $\Omega_j$ is bounded.
Using this fact, the proof of the continuity in Lemma \ref{BBpvolsamecontLEM} carries over to the present situation.
\end{remark}

Concerning the distribution of the limit variables $(\ttau_1,\ldots,\ttau_N)$,
we see from the above proof of \eqref{thm:main001res2} that for any $T_1,\ldots,T_n>0$,
\begin{align}\label{taundistrEXPL1}
&\PP\bigl(\ttau_n\leq T_n\hspace{8pt}\forall n\in\oN\bigr)
=(\mu_X\times\lambda)\bigl(B[(0)_{n=1}^N,(T_n)_{n=1}^N]\bigr).
\end{align}
Combining this with \eqref{BBpvolsamecontLEMpf1} and \eqref{tAjYZ}, we get
\begin{align}\label{taundistrEXPL}
&\PP\bigl(\ttau_n\leq T_n\hspace{8pt}\forall n\in\oN\bigr)
\\\notag
&=(\mu_X\times\lambda)\biggl(\biggl\{(g\Gamma,\vecJ)\col
\sum_{j=1}^k\#\biggl\{\cmatr t{\vecx}\in g^{[j]}(\Z^d)\col
0<t\leq T_n,\: \vecx\in-\tOmega_j(\vecJ)\biggr\}\geq n
\hspace{8pt}\forall n\in\oN\biggr\}\biggr).
\end{align}
Hence the limit variables $(\ttau_i)_{i=1}^\infty$ may be described as follows.
Recall \eqref{PLOMEGAdef}.
Let $\vecJ$ be a random point in $\scrU$ distributed according to $\lambda$,
and let $g\Gamma$ be a random point in $X$ distributed according to $\mu_X$,
and independent from $\vecJ$.
Then $(\ttau_i)_{i=1}^\infty$ can be taken to be the elements of the random set
\begin{align}\label{ttauiGENDISCR2}
\bigcup_{j=1}^k\scrP(g^{[j]}(\Z^d),\tOmega_j(\vecJ)),
\end{align}
ordered by size.
Similarly, $(\tau_i)_{i=1}^\infty$ can be taken to be the elements of the random set
\begin{align}\label{tauiGENDISCR2}
\bigcup_{j=1}^k\scrP(g^{[j]}(\Z^d),\oOmega_j(\vecJ)),
\end{align}
ordered by size.
This description clearly agrees with the one in \eqref{tauiGENDISCR1} and \eqref{ttauiGENDISCR1}.
Let us also note that it follows from \eqref{taundistrEXPL1} and Lemma \ref{BBpvolsamecontLEM},
and the $\osigma^{(k)}_\lambda$-analogues of these, that
the distribution functions
$\PP\bigl(\tau_n\leq T_n\hspace{8pt}\forall n\in\oN\bigr)$ and
$\PP\bigl(\ttau_n\leq T_n\hspace{8pt}\forall n\in\oN\bigr)$ 
depend continuously on $(T_n)\in\R_{>0}^N$,
as stated in Section \ref{sec:general}.

\end{document}